\documentclass[12pt,a4]{article}

\usepackage{etex}
\usepackage{verbatim}
\usepackage{graphicx}
\usepackage{makeidx}  % allows for indexgeneration
\usepackage{stmaryrd}
\usepackage{amssymb,amsmath,amsthm,amsfonts}
\usepackage[utf8]{inputenc}
\usepackage{mathtools}
\usepackage{calc}

\usepackage{authblk}

\usepackage{geometry}
\usepackage{todonotes}
\usepackage[all]{xy}

\usepackage{float}

\usepackage{tikz}
\usepackage{tikz-qtree}
\usetikzlibrary{shapes.geometric,arrows,calc,decorations.pathmorphing}
\usepackage{wrapfig}
\usepackage[font={scriptsize,it}]{caption}

\usepackage[nodisplayskipstretch]{setspace}
\usepackage{amsmath}
\usepackage{enumitem}

\usepackage{esvect}

\theoremstyle{plain}
\newtheorem{theorem}{Theorem}
\newtheorem{lemma}[theorem]{Lemma}

\newtheorem{proposition}[theorem]{Proposition}

\newtheorem{definition}[theorem]{Definition}

\makeatletter
\def\imod#1{\allowbreak\mkern10mu({\operator@font mod}\,\,#1)}
\makeatother

\makeatletter
\theoremstyle{plain}
\newtheorem*{rep@theorem}{\rep@title}
\newcommand{\newreptheorem}[2]{%
\newenvironment{rep#1}[1]{%
\def\rep@title{#2 \ref{##1}}%
\begin{rep@theorem}}%
{\end{rep@theorem}}}
\newreptheorem{theorem}{Theorem}
\newreptheorem{lemma}{Lemma}
\newreptheorem{proposition}{Proposition}

\newcommand{\bydef}{ \stackrel{\mathrm{def}}{=} }

\newcommand{\Tr}{{\mathcal T}}

\newcommand\node[1]{*+[o]{#1}}

\newcommand\addLabelUL[1]{\ar@{}[]+UR|(1){~\makebox[0pt][l]{$\mathbf{#1}$}}}

\newcommand\addLabelUR[1]{\ar@{}[]+UR|(1){~\makebox[0pt][l]{$\mathbf{#1}$}}}

\newcommand\addLabelDL[1]{\ar@{}[]+UR|(1){~\makebox[0pt][l]{$\mathbf{#1}$}}}

\newcommand\addLabelDR[1]{\ar@{}[]+UR|(1){~\makebox[0pt][l]{$\mathbf{#1}$}}}

\newcommand\addDMD[2]{
	\ar@{-}[]+<#1pt,0pt>;[]+<0pt,#2pt>
	\ar@{-}[]+<0pt,#2pt>;[]+<-#1pt,0pt>
	\ar@{-}[]+<-#1pt,0pt>;[]+<0pt,-#2pt>
	\ar@{-}[]+<0pt,-#2pt>;[]+<#1pt,0pt>
}

\usepackage{parallel} % used for parallel presentation of humongous formulas

\newcommand{\trees}[1]{{\mathcal T}_{#1}}

\newcommand{\ignore}[1]{}

\newcommand{\Bb}{{\mathcal B}}
\usepackage{xr-hyper} 
\usepackage{hyperref} 

\hypersetup{%
  colorlinks=false,% hyperlinks will be black
  linkbordercolor=red,% hyperlink borders will be red
  pdfborderstyle={/S/U/W 1}% border style will be underline of width 1pt
}

\newcommand*{\diameter}{0.2cm}
\newcommand*{\smalldiameter}{0.15cm}

\usepackage{color}

\definecolor{mygreen}{rgb}{0,0.6,0}
\definecolor{mygray}{rgb}{0.5,0.5,0.5}
\definecolor{mymauve}{rgb}{0.58,0,0.82}

\usepackage[T1]{fontenc}
\usepackage[scaled]{beramono}

\usepackage{listings}    
\usepackage{chngcntr}
\newcounter{nalg} % defines algorithm counter for chapter-level
\renewcommand{\thenalg}{\arabic{nalg}} %defines appearance of the algorithm counter
\DeclareCaptionLabelFormat{algocaption}{Algorithm \thenalg} % defines a new caption label as Algorithm x.y

\lstnewenvironment{algorithm}[1][numberbychapter=false] %defines the algorithm listing environment
{   
    \refstepcounter{nalg} %increments algorithm number
    \captionsetup{labelformat=algocaption,labelsep=colon} %defines the caption setup for: it ises label format as the declared caption label above and makes label and caption text to be separated by a ':'
    \lstset{ %this is the stype
        frame=tB,
        numbers=left, 
        numberstyle=\tiny,
        basicstyle=\scriptsize, 
        keywordstyle=\color{black}\bfseries\em,
        keywords={,input, output, return, datatype, function, if, else, for, foreach, while, begin, end, } %add the keywords you want, or load a language as Rubens explains in his comment above.
        numbers=left,
        xleftmargin=.04\textwidth,
        mathescape=true,
        numberbychapter=false,
        #1 % this is to add specific settings to an usage of this environment (for instnce, the caption and referable label)
    }
}{}
\title{On the Problem of Computing the Probability of Regular Sets of Trees}
\author[1]{Henryk Michalewski}
\author[2]{Matteo Mio}
\affil[1]{University of Warsaw, Poland}
\affil[2]{CNRS/ENS-Lyon, France}

\begin{document}

\maketitle              

\begin{abstract}
%We present an algorithm for computing the Lebesgue measure of regular sets of trees definable by \emph{game automata}. We then use this algorithm to show three properties of measure on regular languages. Firstly, we recall that Staiger's theorem guarantees that a regular language of infinite words is of measure $0$ if and only if it is meager and in contrast, we show that there exists a regular set of trees of measure $0$ which is comeager. Secondly, while the measure of regular sets of infinite words is always a rational number, we show that in the case of regular languages of trees the measure is always algebraic but generally not rational. Thirdly, we show that the measure of so-called \emph{game languages} $W_{i,k}$ is $0$ if $k$ is odd and is $1$ otherwise.\\

We consider the problem of computing the probability of regular languages of infinite trees with respect to the natural coin-flipping measure. We propose an algorithm which computes the probability of languages recognizable by \emph{game automata}. In particular this algorithm is applicable to all deterministic automata. We then use the algorithm to prove through examples three properties of measure: (1) there exist regular sets having irrational probability, (2) there exist comeager regular sets having probability $0$ and (3) the probability of \emph{game languages} $W_{i,k}$, from automata theory, is $0$ if $k$ is odd and is $1$ otherwise.

%\keywords{regular languages of trees, probability, meta-parity games}
\end{abstract}

%%%%%%%%%%%%%%%%%%%%%%%%%%%%%%%%%%%%%%%%%%%%%%%%%
%
%  Introduction
%
%%%%%%%%%%%%%%%%%%%%%%%%%%%%%%%%%%%%%%%%%%%%%%%%%

\section{Introduction}\label{sec:introduction}

Regular languages of trees are sets of infinite binary trees, labeled by letters from a finite alphabet $\Sigma$,  definable by a formula of Monadic Second Order (MSO) logic interpreted over the full binary tree \cite{thomas96} or, equivalently, specified by an \emph{alternating tree automaton} \cite{mullerschupp}. 
In this paper we consider the following problem. Suppose a $\Sigma$-labeled tree $t$ is generated by 
labeling each vertex by a randomly and uniformly chosen letter $a\!\in\!\Sigma$. For a given regular language $L$, what is the probability that $t$ belongs to $L$?
%The probabilistic law just described is formally captured by
By probability we mean the standard \emph{coin--flipping probability measure} $\mu$ (see Section \ref{tech_back} for definitions) on the space of $\Sigma$-labeled trees. Hence a precise formulation of our problem is as follows.\\

{\bf Probability Problem:} does there exist an algorithm which for a given regular language of trees $L$ computes the probability $\mu(L)$?\\

A qualitative variant %\footnote{The authors thank Prof. Damian~Niwiński for suggesting the qualitative formulation of the problem during a seminar in {W}arsaw, 2014.} 
of the problem only asks for a decision procedure for the question ``is $\mu(L)\!=\!1$?''. The problem is well posed since it was recently shown in  \cite[Theorem 1]{GMMS2014}  that regular sets of trees are measurable with respect to any Borel measure and thus, in particular, with respect to the coin-flipping measure.
%We came across the Probability Problem while investigating the SAT(isfability) problem for probabilistic temporal logics such as pCTL \cite{BaierKatoenBook,Brazdil2008,HJ94pctl}. This connection is described below in Subsection \ref{related_work_intro}. % The Probability Problem is in our view a simple and natural question, arguably of independent theoretical interest.

%%%%%%%%%%%%%%%%%%%%%%%%%%%%%%%%
%%%%%%%%%%%%%%%%%%%%%%%%%%%%%%%%
%%%%%%%%%%%%%%%%%%%%%%%%%%%%%%%%

\subsection{Main Results}

We give a positive solution to the Probability Problem for a subclass of regular languages.% definable by \emph{game automata}. %Definition in Subsection \ref{alt_aut} and proof at the end of Section \ref{automata2mbp}). The general problem remains open.

\newcommand{\compuGame}{ 
Let $L$ be a regular set of infinite trees recognizable by a game automaton. Then the probability of $L$ is computable and is an algebraic number.} 
\begin{theorem}\label{thm:compuGame} 
\compuGame
\end{theorem}

Game automata \cite{prof_fach,game_auto} are special types of alternating parity tree automata. %\todo{Around here we put a reference to Kiefer's work. Also somehow we should integrate with the subsection Related Work. } % generalizing deterministic automata. 
The class of languages recognizable by game automata includes, beside all \emph{deterministic languages}, other important examples of regular sets. The most notable examples are \emph{game languages} $W_{i,k}$ which play a fundamental role in the study of tree languages with topological methods \cite{NiwinskiArnold2008,GMMS2014}.  Game automata definable languages are, at the present moment, the largest known subclass of regular languages for which the long-standing Mostowski--Rabin index problem\footnote{The Mostowski--Rabin problem: for a given regular language $L$, compute the minimal number of \emph{priorities} required to define $L$ using an alternating parity tree automaton.} %Beyond the class of game automata the problem is still open, see \cite{prof_fach,game_auto,gap} for further details.}
 is known to be decidable (see \cite{prof_fach,game_auto}). Theorem \ref{thm:compuGame} confirms the good algorithmic properties of game automata. At the same time, however, we suspect that generalizing the result of Theorem \ref{thm:compuGame} to arbitrary regular languages might be hard. Some ideas for further research in this direction are discussed in Section \ref{conclusion}. %\todo{Add a Conclusion Section? Still no discussion at all about Approximations. I think a Conclusion Section of a few lines would be useful.}

%\noindent
From Theorem \ref{thm:compuGame} we derive the following propositions (for proofs see Section \ref{failvv}).

%\begin{proposition}
\begin{proposition}
\label{prop_example_1}
There exists a regular language of trees $L$ %$L^\irr$ 
definable by a %game (in fact, 
deterministic automaton such that $L$ has an irrational %(not quadratic) 
probability.
\end{proposition}
%\end{proposition}
%\begin{proposition}
\begin{proposition}
\label{prop_example_2}
There exists a regular language of trees $L$ %$L^\comg$, 
definable by a  deterministic 
%game (in fact, deterministic)  
automaton such that $L$ is comeager and $L$ has probability $0$.
\end{proposition}

These two propositions should be contrasted with known properties of regular languages of infinite words. First, a result of Staiger in \cite{staiger} states that a regular language $L$ of infinite words has coin-flipping measure $0$ if and only if it is of \emph{Baire first category} (or \emph{meager}). Proposition \ref{prop_example_2} shows that this correspondence fails in the context of infinite trees. 
Second,  the coin-flipping measure of a regular language of infinite words is always rational (see, e.g., Theorem 2 of \cite{Chat2004}).  Hence, the probabilistic properties of regular languages of trees seem to be significantly more refined than in the case of languages of $\omega$-words. 

Lastly, we calculate the probability of all game languages $W_{i,k}$ (see \cite{NiwinskiArnold2008, GMMS2014} and Subsection \ref{wik}), a result that might eventually be useful given the importance of game languages in the topological study of regular sets of trees.

\newcommand{\propwik}{ 
For $0\!\leq\!i\!<\!k$, the game language $W_{i,k}$ has probability $0$ if $k$ is odd and $1$ if $k$ is even.} 
\begin{proposition}
\label{prop_example_3}
\propwik
\end{proposition}

\subsection{The Algorithm}

In Section \ref{automata2mbp} we propose Algorithm \ref{alg2} which computes the probability of regular languages recognized by game automata.  %as a solution of Theorem \ref{thm:compuGame} 
Algorithm \ref{alg2} is based on a reduction to \emph{Markov Branching plays} (MBP's): to each game automaton $\mathcal{A}$ we associate a MBP $\mathcal{M}$. The \emph{value} of $\mathcal{M}$ can be computed and corresponds to the probability of the language recognized by $\mathcal{A}$.
%The \emph{value} of the MBP associated with a given game automaton $\mathcal{A}$ will be equal to the probability of the language recognized by $\mathcal{A}$.
% and to the computation of the  of the MBP. %associated \emph{value}. 
This reduction to MBP's is described in Sections \ref{metaparity_sec} and \ref{automata2mbp}.

The notion of MBP, as a special kind of \emph{two-player stochastic meta-parity game} has been introduced by the second author in \cite{MioThesis,MIO2012b} 
%\todo{Minor modificaitons: removed comma, "introduced" instead of "defined", etc.} 
in order to interpret a probabilistic version of the modal $\mu$-calculus. For a given MBP $\mathcal{M}$ having $n$ states, the vector $val\!\in\![0,1]^n$ of values of $\mathcal{M}$ can be expressed as the solution of a system $\mathcal{S}$ of (nested) least and greatest fixed-point equations over the space %\todo{"space" instead of "set".} 
$[0,1]^n$. From $\mathcal{S}$ one can then construct a first order formula $\phi_{\mathcal{S}}(val)$ in the language of real-closed fields having the property that $val$ is the unique tuple of real numbers satisfying $\phi_{\mathcal{S}}$. The  tuple  $val$ can be computed by Tarski's \emph{quantifier elimination} algorithm \cite{Tarski1951} and consists of algebraic numbers.  See Algorithm \ref{alg1} in Section  \ref{metaparity_sec} for a %high-level
 description of the procedure for computing the value of MBP's.
%One can 
%We hope that reader will 

One can find interesting the connections between the machinery of MBP's  (and thus, as mentioned, the probabilistic $\mu$-calculus), the class of languages definable by game automata, the algorithmic problem of computing the probability of regular languages of trees and the usage of Tarski's quantifier elimination procedure.

\subsection{Related Work}\label{related_work_intro}

In \cite{staiger_comp} L.~Staiger presented an algorithm for computing the \emph{Hausdorff measure} of regular sets of $\omega$-words. The method, based on the decomposition of the input language into simpler components, can be adapted to compute the coin--flipping measure of regular sets of $\omega$-words. 
Our research on the coin-flipping measure of regular languages of trees can be seen as a continuation of Staiger's work.

Natural variants of the qualitative version of the Probability Problem, obtained by replacing ``has probability $1$''  by other notions of largeness, are known to have positive solutions: in \cite{niwinski91}  D.~Niwiński described an algorithm which takes as input a regular language of trees $L$ (presented as a Rabin tree automaton) and decides if $L$ is uncountable and, similarly, an algorithm for establishing if a regular language of trees $L$ is comeager can be extracted from the result of \cite{MM2015a}.

\emph{Addendum.} After the submission of this article we have been informed that the Probability Problem has already been implicitly considered in \cite{CDK2012}, although differently phrased as the verification problem for a class of stochastic branching processes. Following our terminology, in \cite{CDK2012} the authors provide an algorithm for computing the probability of regular languages definable by deterministic tree automata. Hence our results can be seen as extending the work of \cite{CDK2012} from deterministic to game-automata definable languages.  %\todo{R2: last modification. I removed "seminal work", but I think the rest should be kept.}

%%%%%%%%%%%%%%%%%%%%%%%%%%%%%%%%%%%%%%%%%%%%%%%%%
%
%  Technical Background
%  - Descriptive set theory,
%    (alternating) automata, game automata
%
%%%%%%%%%%%%%%%%%%%%%%%%%%%%%%%%%%%%%%%%%%%%%%%%%
\section{Background in Topology and Automata Theory}\label{sec:background}
\label{tech_back}
\subsection{Topology and measure}
\label{dst}\label{sec:logicautomata}

In this section we present elementary topological and measure--theoretical notions required in this work. We refer to  \cite{kechris} as a standard reference on the subject.

The set of natural numbers is denoted by $\omega$. A topological space $X$ is \emph{Polish} if it is separable and completely metrizable. An important example of a Polish space is the \emph{Cantor space} $\{0,1\}^\omega$ of infinite sequences of bits endowed with the product topology. %The Cantor space is \emph{zero-dimensional}, i.e., it has a basis of \emph{clopen} (both open and closed) sets.
%Given a Polish space $X$ with a fixed metric $d$, a set  $A\!\subseteq X$ is of {\em measure $0$} if for every $\epsilon$ there exist points $x_n\subset X$ and radiuses $r_n>0$  ($n=1,2,\ldots$)  such that $\sum_{n=1}^\infty r_n<\epsilon$ and $A\subset\bigcup_{n=1}^\infty B(x_n,r_n)$, where $B(x_n,r_n)$ stands for a ball centered in $x_n$ of radius $r_n$.  
%A set $B\!\subseteq\! X$ is {\em measurable} if $X\!=\!B\cup M$, for some Borel set $B\!\subseteq\!X$ and a measure $0$ set $M\!\subseteq\!X$. 
In this paper we are 
%mostly 
interested in the probability Lebesgue measure $\mu$ on the product space $\Sigma^I$ for $I$, a countable set of indices. The measure $\mu$ is uniquely defined by the assignment $\mu(\{ t\in \Sigma^I\mid t(i_1) = a_1,\ldots,t(i_k)=a_k\})=(\frac{1}{|\Sigma|})^k$
for $i_1,\ldots,i_k\in I$, $a_1,\ldots,a_k\in\Sigma$  ($i_{j}\neq i_{j'}$ whenever $j\neq j'$, see \cite[Chapter 17]{kechris} for additional details). In particular, for the alphabet $\Sigma=\{0,1\}$ and $I=\omega$ this is known as the coin--flipping probability measure on the Cantor space.

%\todo{I added this paragraph about trees}
The countable set $V\!=\!\{L,R\}^*$ of finite words over the alphabet $\{L,R\}$ is called the \emph{full binary tree} and each $v\!\in\!\{L,R\}^*$ is referred to as a \emph{vertex}. The product space $\Sigma^V$ is denoted by $\trees{\Sigma}$ and an element $t\!\in\!\trees{\Sigma}$ is called a $\Sigma$-labeled tree, or just a $\Sigma$-tree. Intuitively, the stochastic processes associated with the coin--flipping measure $\mu$ on $\trees{\Sigma}$ generates an infinite $\Sigma$-tree by labeling each vertex with a randomly (uniformly) chosen label in $\Sigma$.

Given a topological space $X$, a set  $A\!\subseteq X$ is \emph{nowhere dense} if the interior of its closure is the empty set, that is $\texttt{int}(\texttt{cl}(A))\!=\! \emptyset$. A set $A\!\subseteq\!X$ is of \emph{(Baire) first category} (or \emph{meager}) if $A$ can be expressed as a countable union of nowhere dense sets. %A set $A\!\subseteq\!X$ which is not meager is \emph{of the second (Baire) category}. 
The complement of a meager set is called \emph{comeager}.  %A set $B\!\subseteq\! X$ has the \emph{Baire property} if $X\!=\!U\triangle M$, for some open set $U\!\subseteq\!X$ and meager set $M\!\subseteq\!X$, where $\triangle$ is the operation of symmetric difference $X\triangle Y\!=\! (X\cup Y)\setminus (X\cap Y)$.

\subsection{Alternating Parity Tree Automata and Game Automata}
\label{alt_aut} 
%Among the serval types of tree automata that can be equivalently used to define the class of regular languages of trees, we consider \emph{alternating parity tree automata}.  
We include a brief exposition of alternating automata which follows the presentation in \cite[Appendix C]{mullerschupp}. In this paper we are mostly interested in a subclass of alternating parity tree automata called game automata, which is introduced later in the Section. 
\begin{definition}[Alternating Parity Tree Automaton]
Given a finite set $X$, we denote with $\mathcal{DL}(X)$ the set of expressions $e$ generated by the grammar  $e::= x\!\in\! X \mid e \wedge e \mid e \vee e$. An \emph{alternating parity tree automaton} over a finite alphabet $\Sigma$ is a tuple $\mathcal{A}\!=\!\langle \Sigma,Q, q_0, \delta, \pi)$ where $Q$ is a finite \emph{set of states}, $q_0\!\in\!Q$ is the \emph{initial state}, $\delta: Q\times \Sigma \rightarrow \mathcal{DL}( \{L,R\}\times Q)$ is the \emph{alternating transition function}, and $\pi\!:\! Q\!\rightarrow\omega$ is the \emph{parity condition}.
\end{definition}

An alternating parity tree automaton $\mathcal{A}$ over the alphabet $\Sigma$ defines, or ``accepts'', a set of $\Sigma$-trees. The {\em acceptance} of a tree $t\!\in\! \trees{\Sigma}$ is defined via a two-player ($\exists$ and $\forall$) game of infinite duration denoted by $\mathcal{A}(t)$. Game states of  $\mathcal{A}(t)$ are of the form $\langle \vv{x},q\rangle$ or $\langle \vv{x},e\rangle$ with $\vv{x}\!\in\! \{L,R\}^{*}$, $q\!\in\! Q$ and $e\!\in\! \mathcal{DL}( \{L,R\}\times Q)$. 

The game $\mathcal{A}(t)$ starts at state $\langle\epsilon, q_0\rangle$. Game states of the form $\langle \vv{x},q\rangle$, including the initial state,  have only one successor state, to which the game progresses automatically. The successor state is $\langle\vv{x},e\rangle$ with $e\!=\!\delta(q,a)$, where  $a\!=\! t(\vv{x})$ is the labeling of the vertex $\vv{x}$ given by $t$.  The dynamics of the game %\todo{For the moment I commented out the reference to the picture in Appendix A.} %\footnote{see Figure \ref{gametree} in the Appendix.}
 at states $\langle\vv{x},e\rangle$ depends on the possible shapes of $e$. 
If $e\!=\!e_1\vee e_2$, then Player $\exists $ moves either to $\langle\vv{x},e_1\rangle$ or $\langle\vv{x},e_2\rangle$. If $e\!=\!e_1\wedge e_2$, then Player $\forall$ moves either to $\langle\vv{x},e_1\rangle$ or $\langle\vv{x},e_2\rangle$. If $e\!=\! (L,q)$ then the game progresses automatically to the state $\langle \vv{x}.L, q\rangle$. Lastly, if $e\!=\! (R,q)$ the game progresses automatically to the state $\langle \vv{x}.R, q\rangle$. 
Thus a play in the game $\mathcal{A}(t)$ is a sequence $\Pi$ of game--states, that looks like: $
\Pi = (\langle \epsilon,q_0\rangle, \dots, \langle L, q_1\rangle, \dots,\langle LR, q_2\rangle,\dots,  \langle LRL, q_3\rangle, \dots, \langle LRLL, q_4\rangle,\dots)$, 
where the dots represent part of the play in game--states of the form $\langle\vv{x},e\rangle$. Let  $\infty(\Pi)$ be the set of automata states $q\!\in\! Q$ occurring infinitely often in configurations $\langle\vv{x},q\rangle$ of $\Pi$. We then say that the play $\Pi$ of $\mathcal{A}(t)$ is winning for $\exists$, if $\max\{ \pi(q)\mid q\!\in\!\infty(\Pi)\}$ is an even number. The play $\Pi$ is winning for $\forall$ otherwise. The set (or ``language'') of $\Sigma$-trees defined by $\mathcal{A}$ is the collection $\{t\!\in\!\trees{\Sigma} \mid \exists \textnormal{ has a winning strategy in the game } \mathcal{A}(t)\}$. 

We reserve the symbols $\top$ and $\bot$ for two special \emph{sink} states  having even and odd priority, respectively. The transition function is defined, for all $a\!\in\! \Sigma$, as $\delta(\top,a)\!=\!(L,\top)\wedge (R,\top)$ and  $\delta(\bot,a)\!=\!( L,\bot)\wedge (R,\bot)$. Clearly every tree is accepted at the state $\top$ and rejected at $\bot$. 

{\em Game automata} are a subfamily of alternating parity tree automata satisfying the constraint that, for each $q\!\in\!Q$ and $a\!\in\!\Sigma$, the transition $\delta(q,a)\!=\!e$ has either the form $e\!=\! (L,q_L ) \vee (R,q_R)$ or $e\!=\! (L,q_L ) \wedge (R,q_R)$ (see \cite{prof_fach,game_auto} for more information about this class of automata). Transitions of a game automaton ${\mathcal A}$ can be schematically depicted as in the figure above % \ref{game_auto_pic} 
with the left--hand and right--hand diagrams representing the transitions $(q,a)\to (L,q_L)\wedge (R,q_R)$ and  $(q,a)\to (L,q_L)\vee (R,q_R)$, respectively.

%\tikzset{>={Latex[width=3mm,length=3mm]}}
%\begin{wrapfigure}[10]{R}{0.4\textwidth}
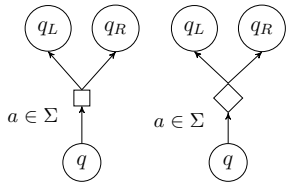
\begin{figure}[H]
\centering
%\vspace{-10pt}
\begin{tikzpicture}[scale=0.7,->,level distance=0.5in,sibling distance=.15in,>=stealth',thick]
%\tikzset{every tree node/.style={align=center,anchor=east,{font=\small}}}
%\tikzset{every tree child node/.style={align=center,anchor=west}}

\tikzset{grow'=up}
\Tree 
[.\node [draw, circle, minimum size = 2*\diameter] {$q$};
        \edge[{font=\small}] node[left = 0.3cm, near end] {$a\in\Sigma$}; 
						    [.\node [draw, minimum size = 2*\smalldiameter] {};
						        \node [draw, circle, minimum size = \diameter] {$q_L$};
						        \node [draw, circle, minimum size = \diameter] {$q_R$};
						    ]
]

\begin{scope}[shift={(2.75cm,0cm)}]
\Tree 
[.\node [draw, circle, minimum size = 2*\diameter] {$q$};
        \edge[{font=\small}] node[left = 0.3cm, near end] {$a\in\Sigma$}; 
						    [.\node [draw, diamond, minimum size = 2*\smalldiameter] {};
						        \node [draw, circle, minimum size = \diameter] {$q_L$};
						        \node [draw, circle, minimum size = \diameter] {$q_R$};
						    ]
]						    
\end{scope}

\end{tikzpicture}
%\vspace{-5pt}
%\caption{Transitions $\delta(q,a)$ of a game automaton.}  %The state $q_2$ is reached in the left son and the state $q_3$ is readched in the right son.}
%\label{game_auto_pic}
\end{figure}
%\end{wrapfigure}
%\todo{I contemplate dropping letter {\bf s} for states in favor of using only $q,q_1,\ldots$; the same applies to the expression in FO(R).}

{\em Deterministic automata} are a subfamily of game automata satisfying the stronger constraint that, for each $q\!\in\!Q$ and $a\!\in\!\Sigma$, the transition $\delta(q,a)\!=\!e$ has the form $e\!=\! (L,q_L ) \wedge (R,q_R)$. Note that the sink states $\top$ and $\bot$ defined above have transitions satisfying this requirement.

%\todo{{\bf 1.} There is a cosmetic difference comparing to the original paper \cite{game_auto} --- we consider one transition less (this looks harmless). {\bf 2.} Their definition of the alternating automaton seems to be identical to ours. }
%\todo[inline]{Remove Banach--Mazur, MSO, include Facchini's description of game automata. Add two pictures about transitions in game--automata. Probably remove alternating automata to Section 7 with only alternating non--game automaton (now I think that probably we should not do so --- we want to have the setup of alternating automata in order to have a one line definition of alternating automata; this is how it is done in \cite{game_auto}). Discuss/advertise examples in future sections.}

%%%%%%%%%%%%%%%%%%%%%%%%%%%%%%%%%%%%%%%%%%%%%%%%%
%
%  Metaparity
%
%%%%%%%%%%%%%%%%%%%%%%%%%%%%%%%%%%%%%%%%%%%%%%%%%
\section{Introduction to meta-parity games}
\label{metaparity_sec}
\label{section:metaparity}
%\emph{Two-player stochastic meta-parity tree games}, or simply $2\frac{1}{2}$-player meta-parity game, are infinite duration games played by Player $1$, Player $2$ and a third probabilistic agent named \emph{Nature}, on a \emph{game arena} $\mathcal{M}\!=\!\langle (S,E), (S_{1},S_{2},S_{N}, B_1,B_2), \pi \rangle$, where $(S,E)$ is a directed graph with finite set of vertices $S$ and transition relation $E$,  $(S_{1},S_{2},S_{N},B_\exists,B_\forall)$ is a partition of $S$ and $\pi\! :\! S_{N}\!\rightarrow\! \mathcal{D}(S)$, where $\mathcal{D}(S)$ denotes the set of discrete probability distributions $d\!:\! S\rightarrow[0,1]$ on $S$. The states in $S_{1}$, $S_{2}$, $S_{N}$ are called \emph{Player $1$} states, \emph{Player $2$} states, \emph{probabilistic} states and \emph{branching} states respectively. The states in $B_\exists$ and $B_\forall$ are called \emph{branching} states.

In this Section we describe a class of stochastic processes called \emph{Markov branching plays} (MBP's) \cite{MioThesis,MIO2012b}  which, as we will observe, is closely related to game automata and will provide a method for calculating the probability of regular languages defined by such automata.  
For a quick overview, a procedure for computing the value associated with a MBP is presented as Algorithm \ref{alg1}, at the end of this section. The procedure for computing the probability  of regular languages defined by game automata in presented as Algorithm \ref{alg2} in the next section.

%We assume familiarity with the standard concepts of two--player parity games %\todo{There is a pretty big jump from what is before in the paper and what happens in this section --- a motivating paragraph would be welcome in my opinion} 
%and two-player stochastic ($2\frac{1}{2}$-player)  parity games. We refer to \cite{ChatPhD} for an excellent introduction to the subject.  

We assume familiarity with the standard concepts of Markov chain and %\todo{There is a pretty big jump from what is before in the paper and what happens in this section --- a motivating paragraph would be welcome in my opinion} 
 two-player stochastic ($2\frac{1}{2}$-player) parity game (see, e.g.,  \cite{ChatPhD}).  Ordinary $2\frac{1}{2}$-player parity games are played on directed graphs whose set of states is partitioned into Player $1$, Player $2$ and probabilistic states. A $2\frac{1}{2}$-player parity game with neither Player $1$ nor Player $2$ states can be identified with a \emph{Markov chain}.

%****************************
%
% Pieces commented by Henryk
% saved for a journal version
%
%****************************
 \begin{comment}
Recall that ordinary $2\frac{1}{2}$-player parity games  are played on directed graphs whose set of states is partitioned into Player $1$ and Player $2$ states, where the two players make their moves, and probabilistic states, where the game progresses to a successor state following a prescribed probabilistic rule. A $2\frac{1}{2}$-player parity game without Player $2$ states is called a \emph{Markov Decision Process} (MDP). A $2\frac{1}{2}$-player parity game with neither Player $1$ nor Player $2$ states (i.e., only having probabilistic states) can be identified with a \emph{Markov chain}. \todo{Here I uncommented certain piece which was reserved for the journal version.}
\end{comment}

\emph{Two-player stochastic} \emph{meta-parity games} \cite{MioThesis,MIO2012b} generalize $2\frac{1}{2}$-player parity games by allowing the directed graph to have two additional kinds of states called $\exists$--branching states and $\forall$--branching states. In this paper we will only consider $2\frac{1}{2}$-player meta-parity games with neither Player $1$ nor Player $2$ states. Such structures, which thus constitute a generalization of Markov chains, are called \emph{Markov branching plays} (MBP's). In what follows we provide a quick description of MBP and refer to \cite{MioThesis} for a detailed account.

\begin{definition}[Markov Branching Play]
A Markov branching play (MBP) is a structure $\mathcal{M}\!=\!\langle (S,E), (S_{P}, B_\exists,B_\forall), p , Par \rangle$ where:
\begin{itemize}
\item  $(S,E)$ is a directed graph with finite set of vertices $S$ and transition relation $E$. We say that $s^\prime$ is a successor of $s$ if $(s,s^\prime)\!\in\! E$. We assume that each vertex has at least one successor state in the graph  $(S,E)$. 
\item The triple $(S_{P},B_\exists,B_\forall)$ is a partition of $S$ into probabilistic, $\exists$-branching and $\forall$-branching states.
\item The function $p\! :\! S_{P}\!\rightarrow\! (S\rightarrow[0,1])$ associates to each probabilistic state $s$ a discrete probability distribution $p(s):S\rightarrow[0,1]$ supported over the (nonempty) set of successors of $s$ in the graph $(S,E)$. 
\item Lastly, the function $Par\!:\! S\rightarrow\omega$ is the \emph{parity (or priority) assignment}. 
\end{itemize}\end{definition}

%****************************
%
% Pieces commented by Henryk
% saved for a journal version
%
%****************************

\begin{comment}
In what follows we identify Markov chains with MBP's without $\exists$--branching nor $\forall$--branching states, i.e., such that $B_\exists\!=\!B_\forall\!=\!\emptyset$.

As usual, a Markov chain $\mathcal{M}$ represents the stochastic process associated with a \emph{random infinite walk} on it's set of states, i.e., an infinite sequence of states randomly generated in accordance with the probability transition function $p$. 
\end{comment}

Recall that a Markov chain represents the stochastic process associated with a \emph{random infinite walk} on its set of states. A MBP represents the more involved stochastic process, described below, of generation of a random unranked and unordered \emph{tree} $T$ whose vertices are labeled by states of the MDP. %\todo{This footnote can stay: I still consider keeping some appendices for the sake of calculations od determinany and some qepcad code. However, this picture may be worthy of including into the paper? What do you think?}\footnote{See Appendix \ref{picgen} for a Figure showing first few steps of the generation process.}. 

\vspace{5mm}
\noindent
\emph{MBP's as Stochastic Processes:} given a MBP $\mathcal{M}\!=\!\langle (S,E), (S_{P}, B_\exists,B_\forall), p , Par \rangle$ and an initial vertex $s_0\!\in\! S$, the stochastic process of construction of $T$ is described as follows. 
\begin{itemize}
\item 
The construction starts from the root of $T$ which is labeled by $s_0$. 
\item A leaf $x$ in the so far constructed tree $T$ is extended, \emph{independently} from all other leaves, depending on the type of its labeling state $s$, as follows:
\begin{itemize}
\item If $s\!\in\! S_P$ then $x$ is extended with a unique child which is labeled by a successor state $s^\prime$ of $s$ randomly chosen in accordance with $p(s)$. 
\item If $s\!\in\! B_\exists$ or $s\!\in\! B_\forall$ and $\{s_1,\dots, s_n\}$ are the successors of $s$ in $\mathcal{M}$, then $x$ is extended with $n$ children $y_1,\dots y_n$ and $y_i$ is labeled by $s_i$, for $1\!\leq\! i\! \leq \! n$. 
\end{itemize}
\end{itemize}

%\section{Pictorial representation of tree generation process for a given MBP}
%\label{picgen}
We give in Figure \ref{matteo_pic1} an example of a MBP. Probabilistic states, $\exists$-branching and $\forall$-branching states are marked as circles, diamond and boxes, respectively.  The first six initial steps of the stochastic process associated with $\mathcal{M}$ at state $q_1$ are depicted in Figure \ref{matteo_pic3}.
In the first step, the construction of $T$ starts by labeling the root by $q_1$. Since $q_1$ is a probabilistic state, the tree is extended (second step) with only one child labeled by either $q_2$ (with probability $\frac{1}{3}$) or $q_3$ (prob. $\frac{2}{3})$. The picture  shows the case when $q_2$ is chosen. Since the new leaf is labeled by $q_2$, and this is a $\exists$-branching state, the tree is extended by adding one new vertex for each successor of $q_2$ in $\mathcal{M}$, i.e., for both $q_1$ and $q_4$. The construction continues as described above. For example, the probability that the generated infinite tree will have the prefix as at the bottom right of Figure \ref{matteo_pic3} is $\frac{1}{3}\cdot \frac{2}{3}\cdot \frac{1}{2}\!=\!\frac{2}{18}$.
\begin{figure}[H] 
%\begin{wrapfigure}[10]{R}{0.3\textwidth}
\centering
%\vspace{-25pt}
\begin{tikzpicture}[scale=.2,->,level distance=0.7in,sibling distance=.4in]

%\draw[gray,step=1] (-10,-2) grid (10, 12);
\clip (-10,-2) rectangle (10, 17);
%\tikzset{grow'=up}
%\Tree 
%[.\node (q1) [draw, circle, minimum size = 4*\diameter] {$q_1$};
%        \edge[{font=\small}] node[right = 0.05cm, near end] {$\frac{1}{3}$};
%        \node (q2) [draw, diamond, aspect = 1.5, minimum size = 4*\diameter] {$q_2$};
%        \edge[{font=\small}] node[left = 0.05cm, near end] {$\frac{2}{3}$};
%        \node (q3) [draw, minimum size = 4*\diameter] {$q_3$};
%]

%\begin{scope}
\node (q1) [draw, circle, minimum size = \diameter,inner sep=2] {$q_1$};
\node (q2) [draw, diamond, aspect = 1.5, minimum size = 2*\diameter,inner sep=1] at ($(q1)+(-5,5)$) {$q_2$};
\node (q3) [draw, minimum size = 1.5*\diameter,inner sep=4] at ($(q1)+(5,5)$) {$q_3$};
\node (q4) [draw, circle, minimum size = \diameter,{font=\small},inner sep=2]  at ($(q1)+(0,10)$) {$q_4$};
\path (q1) edge [{font=\small}] node[right = 0.05cm, near end] {$\frac{1}{3}$} (q2);
\path (q1) edge [{font=\small}] node[left = 0.05cm, near end] {$\frac{2}{3}$} (q3);
\path (q2) edge [bend right=40] node {} (q1);
\path (q2) edge node {} (q4);
\path (q3) edge [bend left=40] node {} (q1);
\path (q3) edge node {} (q4);
\path (q4) edge [font=\small, loop above] node {$\frac{1}{2}$} (q4);
\path (q4) edge [font=\small, left = 0.05cm, near end, bend right=40] node {$\frac{1}{2}$} (q2);
%\end{scope}

\end{tikzpicture}

%\vspace{5pt}
        
\caption{An example of a MBP.}
\label{matteo_pic1}
%\end{wrapfigure}
\end{figure}
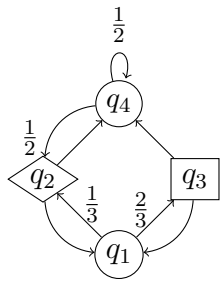

The kind of infinite trees produced by the stochastic process just described %\footnote{The stochastic processes associated with MBP's  are essentially equivalent to the \emph{stochastic branching processes} definable by the grammars considered in \cite{CDK2012}.} 
are called \emph{branching plays}. Branching plays are characterized by the property that each vertex labeled with a probabilistic state has only one child, and each vertex labeled with a ($\exists$ or $\forall$) branching state $s$ has as many children as there are successors of $s$ in the MBP. %Note that since we have assumed that each state in the MBP has at least one successor, there are no leaves in a branching play. 
\begin{figure}[H]
%\begin{wrapfigure}[10]{R}{0.5\textwidth}
\centering
%\vspace{-35pt}
\begin{tikzpicture}[scale=0.7,->,level distance=0.4in,sibling distance=.15in]

\tikzset{grow'=up}
\Tree 
[.\node [draw, circle, minimum size = 2*\diameter] {$q_1$};
        \edge[{font=\small}] node[right = 0.05cm, near end] {$ $};
        [.\node [draw, diamond, aspect = 1.5, minimum size = 2*\diameter] {$q_2$};
            [.\node [draw, circle, minimum size = 2*\diameter] {$q_1$};
                \node [draw, minimum size = 4*\diameter] {$q_3$};
            ]
            \node [draw, circle, minimum size = 2*\diameter] {$q_4$};
        ]
]

\begin{scope}[shift={(5cm,0cm)}]
\Tree 
[.\node [draw, circle, minimum size = 2*\diameter] {$q_1$};
        \edge[{font=\small}] node[right = 0.05cm, near end] {$ $};
        [.\node [draw, diamond, aspect = 1.5, minimum size = 2*\diameter] {$q_2$};
            [.\node [draw, circle, minimum size = 2*\diameter] {$q_1$};
                \node [draw, minimum size = 4*\diameter] {$q_3$};
            ]
            [.\node [draw, circle, minimum size = 2*\diameter] {$q_4$};
                \node [draw, circle, minimum size = 2*\diameter] {$q_4$};
            ]
        ]
]
\end{scope}

\begin{scope}[shift={(10cm,0cm)}]
\Tree 
[.\node [draw, circle, minimum size = 2*\diameter] {$q_1$};
        \edge[{font=\small}] node[right = 0.05cm, near end] {$ $};
        [.\node [draw, diamond, aspect = 1.5, minimum size = 2*\diameter] {$q_2$};
            [.\node [draw, circle, minimum size = 2*\diameter] {$q_1$};
                [.\node [draw, minimum size = 4*\diameter] {$q_3$};
                    \node (q1) [draw, circle, minimum size = 2*\diameter] {$q_1$};
                    \node (q4) [draw, circle, minimum size = 2*\diameter] {$q_4$};
                ]
            ]
            [.\node [draw, circle, minimum size = 2*\diameter] {$q_4$};
                \node (q4prime) [draw, circle, minimum size = 2*\diameter] {$q_4$};
            ]
        ]
]
\end{scope}

\begin{scope}[shift={(0cm,5.5cm)}]
\Tree 
[.\node [draw, circle, minimum size = 2*\diameter] {$q_1$};
]
\end{scope}

\begin{scope}[shift={(5cm,5.5cm)}]
\Tree 
[.\node [draw, circle, minimum size = 2*\diameter] {$q_1$};
        \edge[{font=\small}] node[right = 0.05cm, near end] {$ $};
        \node [draw, diamond, aspect = 1.5, minimum size = 2*\diameter] {$q_2$};
]
\end{scope}

\begin{scope}[shift={(10cm,5.5c m)}]
\Tree 
[.\node [draw, circle, minimum size = 2*\diameter] {$q_1$};
        \edge[{font=\small}] node[right = 0.05cm, near end] {$ $};
        [.\node [draw, diamond, aspect = 1.5, minimum size = 2*\diameter] {$q_2$};
            \node [draw, circle, minimum size = 2*\diameter] {$q_1$};
            \node [draw, circle, minimum size = 2*\diameter] {$q_4$};
        ]
]
\end{scope}

\node (q1v1) at ($(q1)+(-0.5,1)$) {};
\node (q1v2) at ($(q1)+(0.5,1)$) {};
\node (q4v1) at ($(q4)+(0,1)$) {};
\node (q4primev1) at ($(q4prime)+(0,1)$) {};
\path[draw, densely dotted] (q1) edge node {} (q1v1);
\path[draw, densely dotted] (q1) edge node {} (q1v2);
\path[draw, densely dotted] (q4) edge node {} (q4v1);
\path[draw, densely dotted] (q4prime) edge node {} (q4primev1);
\end{tikzpicture}
\caption{The stochastic process associated with the MBP in Figure \ref{matteo_pic1}.}
\label{matteo_pic3}
\end{figure}
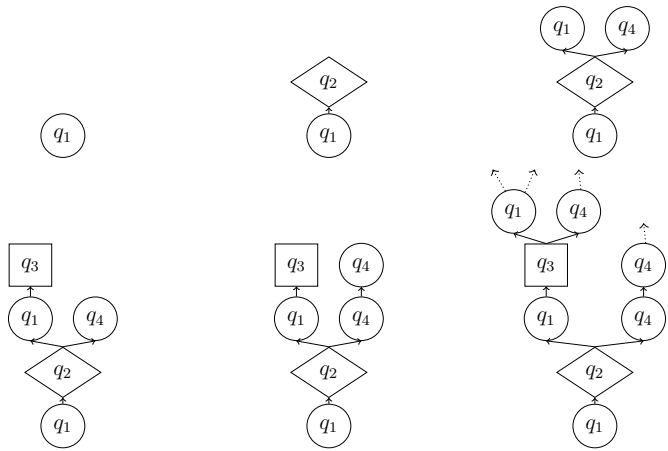
%\end{wrapfigure}

%\todo{I tried to remove these pictures to save space, I'm not sure they are absolutely necessary. Surely the are useful but we cannot have anything...}
%\todo{Replace $q$'s by $s$'s.}
%\vspace{-3cm}
%\input{pictures/matteo_pic1}
%An example of MBP is illustrated on the left by depicting probabilistic states, $\exists$-branching and $\forall$-branching states as circles, diamond and boxes, respectively. Some possible initial steps of the stochastic process having probabilities $(1,\frac{1}{3},1, \frac{2}{3}, \frac{1}{2},1)$ , respectively, are depicted below. Since, each leaf is extended independently from the others, the probability of the process beginning as depicted is $\frac{2}{18}$.
%\input{pictures/matteo_pic3}

%****************************
%
% Pieces commented by Henryk
% saved for a journal version
%
%****************************
\begin{comment}
When $B_\exists\!=\!B_\forall\!=\! \emptyset$ (i.e. when $\mathcal{M}$ is a Markov chain) the produced $T$ has only one infinite branch and can, therefore, be identified as a chain. Hence the stochastic process associated with MBP's generalizes that of Markov chains. 
\end{comment}

%\input{pictures/matteo_pic3.tex}
The collection of branching plays in a MBP $\mathcal{M}$ starting from a state $s$ is denoted by $\mathcal{BP}(\mathcal{M},s)$. %Note that if $\mathcal{M}$ is a Markov chain then $\mathcal{BP}(\mathcal{M},s)$ coincide with the set of infinite paths in $\mathcal{M}$ starting from $s$. 
The set $\mathcal{BP}(\mathcal{M},s)$ naturally carries a Polish topology making $\mathcal{BP}(\mathcal{M},s)$ homeomorphic to the Cantor space (see, e.g., Definition 4.4 in \cite{MIO2012b}). The stochastic process associated to a MBP $\mathcal{M}$, specified on the previous page, can be naturally formalized by a probability measure $\mu_{\mathcal{M}}$ over the space $\mathcal{BP}(\mathcal{M},s)$ of branching plays. See also Definition 4.7 in \cite{MIO2012b} for a formal definition.

%This is is naturally defined accordingly with the specification provided in the paragraph \emph{MBP's as Markov processes} and the above example. See also Definition 4.7 in \cite{MIO2012b} for a formal definition.

Each branching play $T$ can itself be viewed as an ordinary (infinite) two-player parity game $\mathcal{G}(T)$, played on the tree structure of $T$,  by interpreting the vertices of $T$ labeled by $\exists$-branching and $\forall$-branching states as under the control of Player $\exists$ and Player $\forall$, respectively. All other states (i.e., those labeled by a probabilistic state) have a unique successor in $T$ to which the game $\mathcal{G}(T)$ progresses automatically. Lastly, the parity condition associated to each vertex corresponds to the parity assigned in $\mathcal{M}$ to the state labeling it.  We denote with $\mathcal{W}_s$ the set of branching plays starting at $s$ and winning for Player $\exists$, i.e., the set defined as: $\mathcal{W}_s = \big\{ T \in \mathcal{BP}(\mathcal{M},s) \mid \textnormal{ Player $\exists$ has a winning strategy in }\mathcal{G}(T) \big\}$.

\begin{definition}[Value of a MBP]
The \emph{value} of a MBP $\mathcal{M}$ at a state $s$, denoted by $val(\mathcal{M},s)$, is the probability of generating a branching play winning for $\exists$ starting the stochastic process from the state $s$. Formally, $val(\mathcal{M},s)= \mu_{\mathcal{M}}(\mathcal{W}_s)$.
\end{definition}

We remark that the above definition is valid because the set $\mathcal{W}_s$ is $\mu$-measurable for every Borel measure $\mu$ on the space $\mathcal{BP}(\mathcal{M},s)$ (\cite{GMMS2014}) and thus also for $\mu_{\mathcal{M}}$. 

\subsection{How to compute the value of a MBP}\label{subsection_howto_compute}
In this subsection we show how the values  $val(\mathcal{M},s)$ can be computed. The algorithm is based on a result of \cite{MioThesis,MIO2012b}, formulated as Theorem \ref{main_theorem_MBP} below, characterizing such values as the solution of an appropriate system of (least and greatest) fixed-point equations. We first formulate Proposition \ref{bellman} exposing a fixed-point property of the value of MBP's. Let us fix a MBP $\mathcal{M}\!=\!\langle (S,E), (S_{P}, B_\exists,B_\forall), p , Par \rangle$ with $S\!=\! \{s_1\dots s_n\}$. 
To improve readability we just write $val_i$ for $val(\mathcal{M},s_i)$ and we denote by $val$ the vector  $val\!=\!(val_i)_{1\leq i \leq n}$ of length $n$. The symbols $\sum$ and $\prod$ denote the usual operations of sum and product on reals. We also use a ``coproduct'' operation defined as  $\coprod_{i\in I}x_i \!= \!  1-\prod_{i\in I }1-x_i$. %Binary coproducts are denoted by $x\cdot y$. 
%Note that $x\odot y\!=\! x+y-xy$. and $x\odot y$, respectively. Note, that $x\odot y\!=\! x+y-xy$.

%\vspace{-10pt}
\begin{proposition}\label{bellman}
  The equality $val \!=\! f(val)$ holds, where  $f\!:\![0,1]^n \!\rightarrow\! [0,1]^n$ is:
%}
%
%\ParallelRText
%{
%\begin{equation*}
%\resizebox{0.6\hsize}{!}{
\begin{center}
$
\big(f\begin{pmatrix} 
x_1\\ 
\vdots\\
x_n
\end{pmatrix}\big)_i = \left\{\begin{array}{l  l}      
								\displaystyle 	\sum_{\{ j \ \mid \ (s_i,s_j)\in E\}} p(s_i)(s_j) \cdot x_j \ \ & \textnormal{if }\ s_i\!\in\! S_{P} \vspace{5pt} \\
								\displaystyle \prod_{\{ j\  \mid \ (s_i,s_j)\in E\}} x_j & \textnormal{if } \  s_i\!\in\! B_{\forall} \vspace{5pt} \\
								\displaystyle \coprod_{\{ j\  \mid\  (s_i,s_j)\in E\}} x_j & \textnormal{if } \  s_i\!\in\! B_{\exists}
                      	      \end{array}      \right.
$
%}
%\end{equation*}
%\end{Parallel}
%\noindent
\end{center}
\end{proposition}

%Proposition \ref{bellman} in an analogue of the standard Bellman equation characterizing the value (i.e., \emph{winning region}) of ordinary parity games (see Definition 4.3.1 in \cite{Rudiments2001}).
\begin{proof} 
Here we sketch the main idea of the argument, for a formal proof, see Theorem 4.22 of \cite{MIO2012b}.
If $s_i$ is a probabilistic state, then $val_i$ is the weighted average of the value of its successors, since the stochastic process associated with the MBP chooses a unique successor $s_j$ of $s_i$ with probability $p(s_i)(s_j)$.  If $s_i$ is a $\forall$-branching state, then $val_i$ is the probability that all \emph{independently} generated subtrees are winning for Player $\exists$ and this is captured by the $\prod$ expression. Similarly, if $s_i$ is a $\exists$-branching state then $val_i$ is the probability that at least one generated subtree is winning for Player $\exists$, as formalized by the $\coprod$ expression.
Hence the vector $val$ is one of the fixed-points of the function $f\!:\![0,1]^n \!\rightarrow\! [0,1]^n$.
\end{proof}
Theorem \ref{main_theorem_MBP} below %\todo{This is a nice narrative, but I would put not here, but in a more visible place, so a reader can see this interdependence. Also, I would mark something as a ``Proof sketch'' of the Proposition, just for the sake of better separation of the narrative.}
%\footnote{Theorem 6.3.10 of \cite{MioThesis} is stated assuming the validity of the set-theoretic axiom $\textnormal{MA}_{\aleph_1}$. As shown in \cite{GMMS2014}, such assumption is not necessary and can thus be dropped.} 
%is analogous to the corresponding result for ordinary parity games (see Section 4.3 of \cite{Rudiments2001}) and it 
refines Proposition \ref{bellman} by identifying $val$ as the unique vector satisfying a system of nested (least and greatest) fixed-point equations. Its formulation closely follows the notation adopted in the textbook \cite[\S 4.3]{Rudiments2001} for presenting a similar result valid for ordinary parity games. To adhere to such notation, we will define a function $g$, a variant of the function $f$ presented above. Let $k\!=\!\max\{ Par(s) \mid s\!\in\! S\}$ and $l \!=\!\min\{ Par(s) \mid s\!\in\! S\}$  be the maximal and minimal priorities used in the MBP, respectively, and let $c\!=\!k-l+1$. %For each $1\!\leq\! i\! 

\begin{definition}\label{definition_of_g}
The function $g\!:\! ([0,1]^n)^{c} \rightarrow [0,1]^n$ is defined as follows:
\begin{center}
%\begin{equation*}
%\resizebox{0.7\hsize}{!}{
$
\big(g\begin{pmatrix} 
x^l_1\\ 
\vdots\\
x^l_n
\end{pmatrix},\dots, \begin{pmatrix} 
x^k_1\\ 
\vdots\\
x^k_n
\end{pmatrix} \big)_i = \left\{      \begin{array}{l  l}      

									\displaystyle 	\sum_{\{ j \ \mid \ (s_i,s_j)\in E\}} p(s_i)(s_j) \cdot x^{Par(s_j)}_j \ \ & \textnormal{if }  s_i\!\in\! S_{P}\\ \vspace{5pt}
						
								\displaystyle \prod_{\{ j\  \mid \ (s_i,s_j)\in E\}} x^{Par(s_j)}_j & \textnormal{if } \  s_i\!\in\! B_{\forall}\\ \vspace{5pt}

								\displaystyle \coprod_{\{ j\  \mid\  (s_i,s_j)\in E\}} x^{Par(s_j)}_j & \textnormal{if } \  s_i\!\in\! B_{\exists}\\ 

                      	      		     \end{array}      \right.
$
%}
%\end{equation*}
\end{center}
\end{definition}

%Note that while $f$ has as input only one vector in $[0,1]^n$ (providing numerical information for each state in the MDP) the function $g$ has $k$ input vectors. 
%Note that 
The function $g$ depends, like the function $f$,  only on $n$ variables $\{ x^{Par(s_1)}_1, \dots, x^{Par(s_n)}_n\}$ appearing in the body of its definition. The input of $g$ can indeed be regarded as the input of $f$ divided into $c$ baskets, where each variable $x_i$ is put in the basket corresponding  to the priority of $s_i$, for $1\!\leq \! i \!\leq\! n$.

The set $[0,1]^n$, equipped with the pointwise order defined as \[(x_1,\dots,x_n)\!\leq\!(y_1,\dots,y_n) \Leftrightarrow \forall i. (x_i\leq y_i),\] is a complete lattice and the function $g$ is monotone with respect to this order in each of its arguments. Hence the Knaster--Tarski theorem ensures the existence of least and greatest points.  We are now ready to state the main result regarding the values of a given MBP. We adopt standard $\mu$-calculus notation (see, e.g., \cite{Rudiments2001} and \cite{MioThesis,MIO2012b}) to express systems of least and greatest fixed-points equations.
\begin{theorem}[{\cite[Theorem 6.4.2]{MioThesis}}]\label{main_theorem_MBP}
\label{corr1}
The following equality holds:\footnote{Theorem 6.4.2 of \cite{MioThesis} actually proves a stronger result valid for arbitrary $2\frac{1}{2}$-player meta-parity games whereas, as mentioned in the beginning of this section, Markov branching plays are $2\frac{1}{2}$-player meta-parity games without Player $1$ and Player $2$ states. Also, Theorem 6.4.2 of \cite{MioThesis} is stated assuming the validity of the set-theoretic axiom $\textnormal{MA}_{\aleph_1}$, but as shown in \cite{GMMS2014}, such assumption is not necessary and can thus be dropped. 
}

\[ 
\begin{pmatrix} 
val_1\\ 
\vdots\\
val_n
\end{pmatrix}
=
\theta_k
\begin{pmatrix} 
x^k_1\\ 
\vdots\\
x^k_n
\end{pmatrix}.
\cdots
.\theta_l
\begin{pmatrix} 
x^l_1\\ 
\vdots\\
x^l_n
\end{pmatrix}.
g(\begin{pmatrix} 
x^l_1\\ 
\vdots\\
x^l_n
\end{pmatrix},
\ldots,
\begin{pmatrix} 
x^k_1\\ 
\vdots\\
x^k_n
\end{pmatrix}) \]
where $\theta_i$, for $l \!\leq\! i\! \leq\! k$ is a least-fixed point operator ($\mu$) if $i$ is an odd number and a greatest-fixed point operator if ($\nu$) if $i$ is even.
\end{theorem}
\begin{proof}
%The function $g\!:\! ([0,1]^n)^{c} \!\rightarrow\! [0,1]^n$ is clearly monotonic in each of its arguments. Hence the existence of fixed-points of $g$ is guaranteed by the Knaster--Tarski fixed-point theorem.
The proof goes by induction on the number of priorities in the MBP $\mathcal{M}$ and by transfinite induction on a rank-function defined on the space of branching plays.  See \cite{MioThesis} for a detailed proof.
\end{proof}

\noindent

The next theorem states that the value of a MBP is computable and is always a vector of algebraic numbers. 
%The proof is by a reduction to the first order theory of real-closed field which is decidable by Tarski's quantifier elimination algorithm. %connects Theorem \ref{main_theorem_MBP} with Tarski's quantifier elimination algorithm. %the above reasoning and in Section \ref{examples_section} we will apply this method to three concrete examples. 
The examples discussed in Section \ref{examples_section} will illustrate the applicability of this result.

\begin{theorem}\label{computability_of_MBP}
Let $\mathcal{M}$ be a MBP. Then for each state $s_i$ of $\mathcal{M}$ the value $val_i$ is computable and is an algebraic number.\end{theorem}
\begin{proof} (sketch)
Using known ideas (see, e.g., Lemma 9 in \cite{AM04} and Proposition 4.1 in \cite{MioSimpsonFICS2013}) the unique vector $val\!=\!(val_1,\dots, val_n)$ satisfying the system of fixed-point expressions $\mathcal{S}$ given by Theorem \ref{main_theorem_MBP} can be computed by a reduction to the first-order theory of real closed fields. 
A first order formula $F(x_1,\dots,x_n)$, inductively defined from $\mathcal{S}$, is constructed with the property that $(val_1,\dots, val_n)\!\in\!\mathbb{R}^n$ is the unique vector of reals satisfying the formula $F(x_1,\dots,x_n)$. By Tarski's quantifier elimination procedure \cite{Tarski1951}, the formula $F(x_1,\dots,x_n)$ can be effectively reduced to an equivalent formula $G(x_1,\dots, x_n)$ without quantifiers, that is, to a Boolean combination of equations and inequalities between polynomials over $(x_1,\dots,x_n)$. It then follows that the  $(val_1,\dots, val_n)$, which can be extracted from $G$ with standard methods, is a vector of algebraic numbers.  In Section \ref{sec:examples} we apply the procedure described above to a number of examples.
\end{proof}

\begin{algorithm}[caption={computing the vector of values of a MBP.}, label={alg1}]
 input: $\textnormal{a Markov Branching Play } {\mathcal M}$. 
 output: $\textnormal{algebraic numbers }r_1,\ldots,r_n\in{\mathbb R} \textnormal{ equal to } (val_1,\ldots,val_n)$. 
 begin
   $\mathcal{S}\gets \textnormal{Generate system of fixed--point equations associated to }\mathcal{M}$
   $F(x_1,\ldots,x_n)\gets \textnormal{Rewrite }\mathcal{S} \textnormal{ to the corresponding first-order formula over\ } FO(\mathbb{R},<,0,1,+,\times) $
   $G(x_1,\ldots,x_n)\gets \textnormal{Apply quantifier elimination procedure to } F(x_1,\ldots,x_n)$
 return $\textnormal{the unique vector }(r_1,\dots, r_n) \textnormal{ satisfying } G(x_1,\dots, x_n)$ 
\end{algorithm}

\section{From Game Automata to Markov Branching Plays}
\label{automata2mbp}
In this section we present a reduction of the problem of computing the probability of regular languages definable by game automata  to the problem of computing the value of a given MBP, which is algorithmically solvable using Algorithm \ref{alg1}.
\begin{comment}
\begin{algorithm}[caption={Translation of automaton $\mathcal{A}$ to MBP $\mathcal{M}$}, label={alg1}]
 input: a game automaton $\mathcal{A}\!=\! ( Q, q_0, \Delta, \pi)$
 output: a Markov Branching Process  $\mathcal{M}\!=\! \langle (S,E), (S_P,B_\exists, B_\forall),p,Par\rangle$
 begin
   /* Initialization - implicitly we initialize 
    *                  all parameters of ${\mathcal M}$ to $\emptyset$ */
   for $q\in Q$
     $S_P\gets S\cup \{s_q\}$
     $Par(s_q)\gets Par(q)$
   for $\delta(q,a)\!=\!(q_L,L)\vee (q_R,R)\in\Delta$
     $B_\exists\gets B_\exists\cup\{s_{q,a}\}$
     $Par(s_{q,a})\gets Par(q)$
     $p(s_q)(s_{q,a})\gets\frac{1}{|\Sigma|}$
     $E\gets E\cup \{ (s_q,s_{q,a}), (s_{q,a},s_{q_L}), (s_{q,a},s_{q_R}) \}$
   for $\delta(q,a)\!=\!(q_L,L)\wedge (q_R,R)\in\Delta$
     $B_\forall\gets B_\forall\cup\{s_{q,a}\}$
     $Par(s_{q,a})\gets Par(q)$
     $p(s_q)(s_{q,a})\gets \frac{1}{|\Sigma|}$
     $E\gets E\cup \{ (s_q,s_{q,a}), (s_{q,a},s_{q_L}), (s_{q,a},s_{q_R}) \}$
 end       
\end{algorithm}
\end{comment}

We now describe how to construct from a game automaton $\mathcal{A}\!=\! ( Q, q_0, \delta, \pi)$ over the alphabet $\Sigma$  a corresponding MBP $\mathcal{M}\!=\! \langle (S,E), (S_P,B_\exists, B_\forall),p,Par\rangle$. 
The set $S$ of states of $\mathcal{M}$ contains a probabilistic state $s_q$, for each $q\!\in\!Q$, a $\exists$--branching state $s_{q,a}$ for each pair $(q,a)$, with $q\!\in\!Q$ and $a\!\in\!\Sigma$, such that $\delta(q,a)\!=\!(L,q_L)\vee (R,q_r)$, and a $\forall$--branching state $s_{q,a}$ for each pair $(q,a)$ such that $\delta(q,a)\!=\!(L,q_L)\wedge (R,q_r)$. The transition relation $E$ is defined as follows:
\begin{itemize}
\item a probabilistic state $s_q$ has as successors the states $\{ s_{q,a} \mid a\!\in\! \Sigma\}$,
\item a $\exists$-branching (resp. $\forall$-branching) state $s_{q,a}$  have two successors $s_{q_1}$ and $s_{q_2}$ where $\delta(q,a)\!=\! (L,q_1)\vee (R,q_2)$ (resp. $\delta(q,a)\!=\! (L,q_1)\wedge (R,q_2)$).
 \end{itemize}
 Note that each state $s_q$, for $q\!\in\! Q$ has exactly $|\Sigma|$ successors and that each state $s_{q,a}$ has exactly\footnote{\label{foot_one} We are implicitly assuming, for the sake of simplicity, that each transition $(L,q_1)\!\wedge\! (R,q_2)$ and $(L,q_1)\!\vee\! (R,q_2)$ of $\delta$ in $\mathcal{A}$ is such that $q_1\!\neq\! q_2$,  and thus that $s_{q,a}$ has exactly two successors.  If necessary, the game-automaton $\mathcal{A}$ can be made satisfy this assumption by introducing additional copies of the states.} two successors. The assignment $p\!:\!S_P\rightarrow(S\rightarrow [0,1])$ is defined as assigning to each probabilistic state (i.e., state of the form $s_q$) a uniform distribution over its successors, that is, $p(s_q)(s_{q,a})\!=\!\frac{1}{|\Sigma|}$. Lastly, the parity assignment $Par\!:\! S\!\rightarrow\!\omega$ of the MBP $\mathcal{M}$ is defined as in the parity condition $\pi$ of the game automaton $\mathcal{A}$ by the mapping $Par(s_q)\!=\!Par(s_{q,a})\!=\! \pi(q)$.

As an illustrative example of this translation, consider the deterministic automaton $\mathcal{A}\!=\! \langle \{q_1,q_2\}, q_1, \delta, \pi\rangle$ over the alphabet $\Sigma\!=\!\{a,b,c\}$, with parity assignment $\pi(q_2)\!=\!2$, $\pi(q_1)\!=\!1$ and transition $\delta$ defined by $\delta(q_1,a)\!=\!\delta(q_2,a)\!=\! (L,q_2)\wedge (R,q_2)$ and $\delta(q_1,l)\!=\!\delta(q_2,l)\!=\! (L,q_1)\wedge (R,q_1)$, for $l\!\in\!\{a,b\}$.

%\tikzset{>={Latex[width=3mm,length=3mm]}}
%\begin{wrapfigure}[10]{L}{0.63\textwidth}
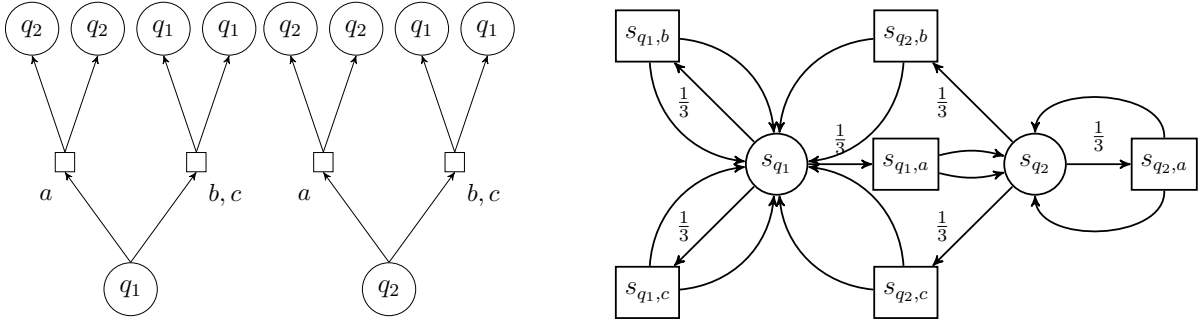
\begin{figure}[H]
%\vspace{10pt}
%\scalebox{0.5}{%
\begin{tikzpicture}[scale=0.85,->,level distance=0.8in,sibling distance=.07in,>=stealth',thick]
\tikzset{grow'=up}
\Tree 
[.\node [draw, circle, minimum size = 2*\diameter] {$q_1$};
        \edge[{font=\small}] node[left = 0.3cm, near end] {$a$}; 
						    [.\node [draw, minimum size = 2*\smalldiameter] {};
						        \node [draw, circle, minimum size = \diameter] {$q_2$};
						        \node [draw, circle, minimum size = \diameter] {$q_2$};
						    ]
        \edge[{font=\small}] node[right = 0.3cm, near end] {$b,c$}; 
				[.\node [draw, minimum size = 2*\smalldiameter] {};
						        \node [draw, circle, minimum size = \diameter] {$q_1$};
						        \node [draw, circle, minimum size = \diameter] {$q_1$};
                ]
]

\pgflowlevelsynccm
\begin{scope}[shift={(4cm,0cm)}]
\Tree 
[.\node [draw, circle, minimum size = 2*\diameter] {$q_2$};
        \edge[{font=\small}] node[left = 0.3cm, near end] {$a$}; 
						    [.\node [draw, minimum size = 2*\smalldiameter] {};
						        \node [draw, circle, minimum size = \diameter] {$q_2$};
						        \node [draw, circle, minimum size = \diameter] {$q_2$};
						    ]
        \edge[{font=\small}] node[right = 0.3cm, near end] {$b,c$}; 
				[.\node [draw, minimum size = 2*\smalldiameter] {};
						        \node [draw, circle, minimum size = \diameter] {$q_1$};
						        \node [draw, circle, minimum size = \diameter] {$q_1$};
                ]
]
\end{scope}

\pgflowlevelsynccm
\begin{scope}[shift={(10cm,0cm)}]
\node (s1) [draw, circle, minimum size = 2*\diameter]  at ($(0,2)$) {$s_{q_1}$};
\node (sq1c) [draw, minimum size = 4*\diameter]  at ($(s1)+(-2,-2)$) {$s_{q_1,c}$};
\node (sq1b) [draw, minimum size = 4*\diameter]  at ($(s1)+(-2,2)$) {$s_{q_1,b}$};
\node (sq1a) [draw, minimum size = 4*\diameter]  at ($(s1)+(2,0)$) {$s_{q_1,a}$};
\node (s2) [draw, circle, minimum size = 2*\diameter]  at ($(s1)+(4,0)$) {$s_{q_2}$};
\node (sq2c) [draw, minimum size = 4*\diameter]  at ($(s2)+(-2,-2)$) {$s_{q_2,c}$};
\node (sq2b) [draw, minimum size = 4*\diameter]  at ($(s2)+(-2,2)$) {$s_{q_2,b}$};
\node (sq2a) [draw, minimum size = 4*\diameter]  at ($(s2)+(2,0)$) {$s_{q_2,a}$};
\path (s1) edge node[{font=\small},left = 0.2cm] {$\frac{1}{3}$} (sq1c);
\path (sq1c) edge [bend right=40] node {} (s1);
\path (sq1c) edge [bend left=40] node {} (s1);
\path (s1) edge [{font=\small},->,left = 0.2cm] node[left=0.2cm] {$\frac{1}{3}$} (sq1b);
\path (sq1b) edge [bend left=40] node {} (s1);
\path (sq1b) edge [bend right=40] node {} (s1);
\path (s1) edge [{font=\small},above = 0.05cm]  node {$\frac{1}{3}$} (sq1a);
\path (sq1a) edge [bend right=15] node {} (s2);
\path (sq1a) edge [bend left=15] node {} (s2);
\path (s2) edge [{font=\small},->,left = 0.2cm] node[left=0.2cm] {$\frac{1}{3}$} (sq2c);
\path (sq2c) edge [bend left=40] node {} (s1);
\path (sq2c) edge [bend right=40] node {} (s1);
\path (s2) edge node[{font=\small},left = 0.2cm] {$\frac{1}{3}$} (sq2b);
\path (sq2b) edge [bend right=40] node {} (s1);
\path (sq2b) edge [bend left=40] node {} (s1);
\path (s2) edge node[{font=\small},above = 0.05cm] {$\frac{1}{3}$} (sq2a);
\path (sq2a) edge [bend left=90] node {} (s2);
\path (sq2a) edge [bend right=90] node {} (s2);
\end{scope}

\pgflowlevelsynccm
\end{tikzpicture}

\caption{Transitions of the game automaton $\mathcal{A}$ and corresponding MBP $\mathcal{M}$.}
\label{gvv_3}
\end{figure}
The corresponding MBP $\mathcal{M}$ is schematically\footnote{\label{foot_two}Due to the chosen succinct definition, the automaton $\mathcal{A}$ does not satisfy the assumption of Footnote \ref{foot_one}. Rather than formally introducing copies $q_1$ and $q_2$ in $\mathcal{A}$, we have simply depicted all $\forall$-branching states of $\mathcal{M}$ as having two successors.} depicted in Figure \ref{gvv_3} (right), %\todo{Put this MDP, having 2 probabilistic states and 3+3 $\forall$-branching states in the same picture, above}
by representing probabilistic states with circles, $\forall$-branching states with boxes and  the probabilistic assignment $p$ by the probabilities labeling the outgoing edges of probabilistic states.
% This general observation, not specific to the proposed example, allows us to state our main correctness theorem.
The soundness of our reduction is stated as follows.
\begin{theorem}[Correctness of Reduction]\label{thm:correctness}
Let $L$ be a regular language recognized by a game automaton $\mathcal{A}$ and let $\mathcal{M}$ be the MBP corresponding to $\mathcal{A}$. Then $\mu(L)\!=\! Val(\mathcal{M},s_{q_0})$, where $q_0$ is the initial state of $\mathcal{A}$. 
\end{theorem}
\begin{proof} (sketch)
Since each probabilistic state has exactly one successor for every letter $a\!\in\! \Sigma$ and each branching state have precisely two successors, there exists a one-to-one correspondence between $\Sigma$-trees $t\!\in\!\trees{\Sigma}$ and branching plays $T\!\in\! \mathcal{BP}(\mathcal{M},s_{q_0})$. Furthermore, it follows directly from the definition of acceptance by $\mathcal{A}$ (see Section \ref{alt_aut}) and  the definition  of the set $\mathcal{W}_{s}$ (see Section \ref{metaparity_sec}) that $t$ is accepted by $\mathcal{A}$ if and only if the corresponding branching play $T$ is in $\mathcal{W}_s$. Lastly, due to the uniform assignment $p$ of probabilities in $\mathcal{M}$, the coin-flipping measure $\mu$ on $\trees{\Sigma}$ and the probability measure $\mu_{\mathcal{M}}$ on  $\mathcal{BP}(\mathcal{M},s_{q_1})$ are identical.
\end{proof}
The result of Theorem \ref{thm:compuGame} in the Introduction then follows as a corollary of Theorem \ref{thm:correctness} above and the fact that the vector of values of a MBP can be computed using Algorithm \ref{alg1}. The final algorithm for computing the probability of regular languages definable by game automata is then as follows.

\begin{algorithm}[caption={computing the probability of regular languages $L$ recognized by game automata.}, label={alg2}]
 input: $\textnormal{a game automaton }\mathcal{A}\!=\! ( Q, q_0, \delta, \pi)$ recognizing a language $L$.
 output: $\textnormal{ a real number corresponding to }\mu(L)$.
 begin
  $\mathcal{M} \gets \textnormal{ Construct the MBP } \mathcal{M} \textnormal{ corresponding to }\mathcal{A}$
  $(val_1,\dots, val_n) \gets \textnormal{Apply Algorithm \ref{alg1} to compute the vector of values of the states of }\mathcal{M}$
 return  $\textnormal{the value }val_{i}\textnormal{ where } i \textnormal{ is the index of the probabilistic state }s_{q_0}\textnormal{ of }\mathcal{M}$
 
\end{algorithm}

\section{Examples}
\label{examples_section}
\label{failvv}
\label{sec:examples}
\label{fail_vv}

In this section we will apply Algorithm \ref{alg2} to analyze examples which will prove Propositions \ref{prop_example_1}, \ref{prop_example_2} and \ref{prop_example_3} stated in the Introduction. In some instances, in order to perform the quantifier elimination procedure required by Algorithm \ref{alg1}, we use the tool %\footnote{Most of the  {\tt qepcad} computations are presented in Subsection \ref{app_failvv}.} 
{\tt qepcad} \cite{tool:qepcad}.

We fix the alphabet $\Sigma\!=\!\{a,b,c\}$ and, for each $n\!\in\!\omega$, we define the regular language $L_n\!\subseteq\!\trees{\Sigma}$ as $L_n = \{ t\! \in\!\trees{\Sigma} \mid \textnormal{$a$ appears $\geq n$ times on every branch of $t$}\}$ and the language $L_{\infty}$ as $L_{\infty}\!=\! \bigcap_{n\in\omega} L_n$, i.e., as the set of $\Sigma$-trees having, on every branch, infinitely many occurrences of the letter $a$. 

\subsection{An introductory example} 
The language $L_1$ is recognized by the deterministic automaton in Figure \ref{gvv_1} (left) defined as $\mathcal{A}_1\!=\! \langle\{q_1,\top\},q_1,\delta_1,\pi\rangle$ where $\top$ is an accepting sink state (see Section \ref{alt_aut} for automata--related definitions), the priority assignment is $\pi(q_1)\!=\!1$ and the transition function $\delta_1$ is defined on $q_1$ as $\delta_1(q_{1},a)\!=\! (L,\top)\wedge(R,\top)$ and $\delta_1(q_{1},l)\!=\! (L,q_1)\wedge(R,q_1)$ for $l\!\in\!\{b,c\}$.
\input{pictures/gvv_picture1} 
We will compute the probability $\mu(L_1)$ using the procedure of Algorithm \ref{alg2}. As a first step we construct the MBP $\mathcal{M}_1$ corresponding to $\mathcal{A}_1$, as specified in Section \ref{automata2mbp}. In order to improve readability, we have represented in Figure \ref{gvv_1} (center)  a simplified version of $\mathcal{M}_1$ where the states $s_{\top}$, $s_{\top,a}$, $s_{\top,b}$ and $s_{\top,c}$ have been identified with the single state $s_{q_1,a}$. This is convenient since, clearly, all of these states have value $1$. Accordingly, the MBP $\mathcal{M}_1$ has four states, all of priority $1$.  Following the procedure of Algorithm \ref{alg2} we need to compute the values of the states of $\mathcal{M}_1$ using Algorithm \ref{alg1}. In accordance with Theorem \ref{main_theorem_MBP}, the fixed-point equation characterizing the vector $val\!=\!(val_{s_{q_1}}, val_{s_{q_1},a}, val_{s_{q_1},b}, val_{s_{q_1},c})$ of values of the states of $\mathcal{M}_1$ is $val \!=\! \mu \vec{x}.g(\vec{x})$, where $g$ is defined as in Figure  \ref{gvv_1} (right). Then $val_{s_{q_1}}$ is the least solution in $[0,1]$ of the equation $x\!=\! \frac{1}{3} + \frac{2}{3}x^{2}$. As it is simple to verify, even without running the solver based on Tarski's quantifier elimination procedure, the solution is $val_{s_{q_1}}\!=\!\frac{1}{2}$, and this is the output returned by Algorithm \ref{alg2}. Hence the probability of $L_1$ is  $\mu(L_1)\!=\!\frac{1}{2}$.

 %In this instance most of the steps of Algorithm 1 trivialize. % as it is simple to calculate even without recurring to the algorithm solver based on the reduction to the first-order theory of reals described in Section \ref{subsection_howto_compute}.

\subsection{Examples of regular languages having irrational probabilities} 
\label{subsection:irrational} This subsection constitutes a proof of Proposition \ref{prop_example_1}. The automaton $\mathcal{A}_2$ recognizing the language $L_2$  is defined as $\mathcal{A}_2\!=\! \langle (\{q_1,q_2, \top\}, q_2, \delta_2,\pi)$ where $q_2$ is the initial state, the priority function is defined as $\pi(q_1)\!=\! \pi(q_2)\!=\! 1$ and the transition function $\delta_2$ is defined on $q_1$ as the function $\delta_1$ of the previous example, and on the state $q_2$ as $\delta(q_2,a)\!=\! (L,q_1)\wedge (R,q_1)$ and $\delta_2(q_2,l)\!=\! (L,q_2)\wedge (R,q_2)$, for $l\!\in\!\{b,c\}$. The transition $\delta_2$ is shown in Figure \ref{gvv_2} (left).
\input{pictures/gvv_picture2}
The MBP $\mathcal{M}_2$ corresponding to $\mathcal{A}_2$ extends the MBP $\mathcal{M}_1$ of the previous example with the probabilistic state $s_{q_2}$ and the three $\forall$-branching states $s_{q_2,a}$, $s_{q_2,b}$ and $s_{q_2,c}$. The new part of the automaton $\mathcal{A}_2$ is depicted in Figure \ref{gvv_2} (center). Noting the four new states are not reachable by the other states already present in $\mathcal{M}_1$, we already know that $val_{s_{q_1}}\!=\! \frac{1}{2}$. Hence we can consider the simplified system of fixed-point equations $\mu \vec{x}.g(\vec{x})$ for calculating the values $val\!=\!(val_{q_2}, val_{q_{2},a}, val_{q_{2},b}, val_{{q_2},c})$ where $g$ is defined in Figure \ref{gvv_2} (right).
Hence the value $val_{q_2}$ is the least solution in $[0,1]$ of the equation $x = \frac{1}{12}+ \frac{2}{3}x^2$ and this is $val_{q_2}\!=\!\frac{1}{4}(3-\sqrt{7})$ which is irrational and approximately equal to $0.088$.

%%%%%%%%%%%%%%%%%%%%%%%%%%%%%%%%%%
%
% Things included from the Appendix. Concern the language L_2
%
%
%%%%%%%%%%%%%%%%%%%%%%%%%%%%%%%%%%

\subsection{Automated computations for $L_2$}
\label{l2}
Above we computed the probability of $L_2$ using elementary ad hoc considerations. Here we will compute again this probability  using Algorithm \ref{alg2} and the tool {\tt qepcad} \cite{tool:qepcad}.  As shown above, % Section \ref{fail_vv},
 the measure of language $L_2$ is described by the following system of fixed--point equations
\[\left\{
\begin{array}{lcl} 
x_1 &  \stackrel{\mu}{=} & \frac{1}{3} + \frac{2}{3}x_1^2\\
x_2 &  \stackrel{\mu}{=} & \frac{1}{3}x_1^2 + \frac{2}{3}x_2^2
\end{array}
\right.
\]
\noindent
This translates to the following {\tt qepcad} code:
\begin{lstlisting}
(A x1prime) 
[
    [x1 = 1/3 + 2/3 x1^2] /\ 
    [x1prime = 1/3 + 2/3 x1prime^2 ==> x1prime >= x1 ]
].
\end{lstlisting}
\vspace{-20pt}
The above formula $\phi_1(x_1)$, written in a formalized language of {\tt qepcad}, expresses that $x_1$ is the smallest solution to the first equation. The solution generated by the tool is $2x_1 - 1 = 0$.
Using this information we get the next formula $\phi_2(x_2)$, saying that $x_2$ is the least solution of the system of fixed point equations satisfying the inequality $x_2\leq 1$.  
\begin{lstlisting}
(E x1)(A x2prime)(A x1prime)
[   
    [2 x1 - 1 = 0] /\ 
    [x2 = 1/3 x1^2 + 2/3 x2^2] /\
    [x2 <= 1] /\
    [
        [
            [2 x1prime - 1 = 0] /\ 
            [x2prime = 1/3 x1prime^2 + 2/3 x2prime^2] /\
            [x2prime <= 1] 
        ] ==> 
        [x2prime >= x2]
     ]
].
\end{lstlisting}
%\vspace{-26pt}

\noindent
The tool {\tt qepcad} reduces the above formula $\phi_2(x_2)$ to a quantifier free expression
\begin{lstlisting}
x2 - 1 < 0 /\ 8 x2^2 - 12 x2 + 1 = 0
\end{lstlisting}

\vspace{-15pt}
\noindent
which in turn can be easily solved analytically. The formula is satisfied by exactly one 
$x = \frac{1}{4}(3-\sqrt{7})$, which is irrational and amounts approximately to $0.088$.

%%%%%%%%%%%%%%%%%%%%%%%%%%%%%%%%%%
%
% Things included from the Appendix. Concern the language L_3
%
%
%%%%%%%%%%%%%%%%%%%%%%%%%%%%%%%%%%

\subsection{Automated computations for $L_3$}
\label{l3}

% One can verify\footnote{See Appendix \ref{l3}.} 
We show below that the probability of $L_3$ is $\mu(L_3)\!=\!\frac{1}{4}(3-\sqrt{1+3\sqrt{7}})$ and thus not of the form $\frac{a+b\sqrt{c}}{d}$ for integers $a,b,c,d$. This means that  $\mu(L_3)$ is not a quadratic irrational. By a characterization proved by Euler and Lagrange this in turn means that the continued fraction representation of $\mu(L_3)$ is not eventually periodic. 

% The language $L_3$ was defined in the first  paragraph of Section \ref{fail_vv}. 
%The language $L_3$ was defined in the first  paragraph of this section. % the third upper approximation of the language $L_\infty$ defined at the beginning of this section. 

The automaton $\mathcal{A}_3$ recognizing the language $L_3$ is defined as: $\mathcal{A}_3\!=\! \langle \{q_1, q_2, q_3, \top\}, q_3, \delta_3,\pi\rangle$ where $q_3$ is the initial state, the priority function is defined as $\pi(q_1)\!=\! \pi(q_2)\!=\!  \pi(q_3) \!= \!1$ and the transition function $\delta_3$ is defined on $q_1$ and $q_2$ as the function $\delta_2$ of the previous example for the language $L_2$, and on the state $q_3$ as $\delta_3(q_3,a)\!=\! (L,q_2)\wedge (R,q_2)$ and $\delta_3(q_3,l)\!=\! (L,q_3)\wedge (R,q_3)$, for $l\!\in\!\{b,c\}$.

One can visualize the new state and transition as shown in Figure \ref{gvv_4}.% \todo{Looks quite ugly. Maybe we can expand as in the case of $L_2$ above.}
%is described in the following figure
%\tikzset{>={Latex[width=3mm,length=3mm]}}
%\begin{wrapfigure}[10]{L}{0.63\textwidth}
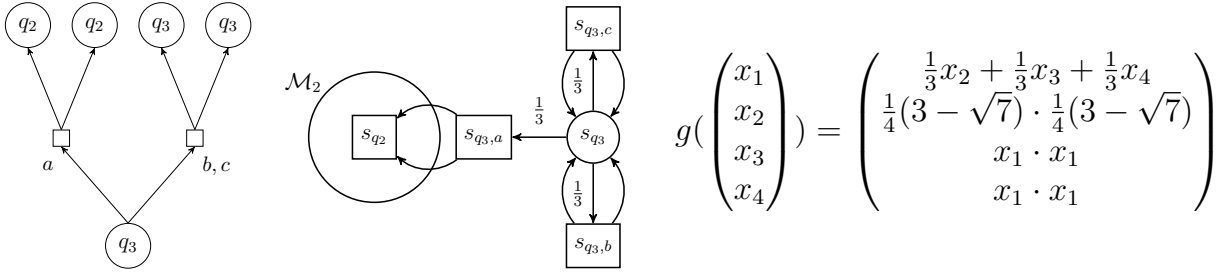
\begin{figure}[H]
%\centering
%\vspace{-35pt}
\begin{tikzpicture}[scale=.72,->,level distance=0.8in,sibling distance=.15in,>=stealth',thick]
%\tikzset{every tree node/.style={align=center,anchor=east,{font=\small}}}
%\tikzset{every tree child node/.style={align=center,anchor=west}}

\tikzset{grow'=up}
\Tree 
[.\node [draw, circle, minimum size = 2*\diameter] {$q_3$};
        \edge[{font=\small}] node[left = 0.3cm, near end] {$a$}; 
						    [.\node [draw, minimum size = 2*\smalldiameter] {};
						        \node [draw, circle, minimum size = \diameter] {$q_2$};
						        \node [draw, circle, minimum size = \diameter] {$q_2$};
						    ]
        \edge[{font=\small}] node[right = 0.3cm, near end] {$b,c$}; 
				[.\node [draw, minimum size = 2*\smalldiameter] {};
						        \node [draw, circle, minimum size = \diameter] {$q_3$};
						        \node [draw, circle, minimum size = \diameter] {$q_3$};
                ]
]

\begin{scope}[shift={(1.5cm,0cm)},>=stealth',thick]
\pgflowlevelsynccm
\node (s1) [draw, minimum size = 4*\diameter] at ($(q2)+(3,2)$) {$s_{q_2}$};
\node (s1big) [draw, circle, minimum size = 12*\diameter] at ($(q2)+(3,2)$) {};
\node (m) [draw=none] at ($(s1)+(-1.3,1)$) {${\mathcal M}_2$};

\node (s2) [draw, circle, minimum size = 2*\diameter] at ($(q2)+(7,2)$) {$s_{q_3}$};
\node (sq2a) [draw, minimum size = 4*\diameter] at ($(s2)+(-2,0)$) {$s_{q_3,a}$};
\node (sq2c) [draw, minimum size = 4*\diameter] at ($(s2)+(0,2)$) {$s_{q_3,c}$};
\node (sq2b) [draw, minimum size = 4*\diameter] at ($(s2)+(0,-2)$) {$s_{q_3,b}$};

\path (s2) edge [left = 0.01cm] node[above = 0.1cm] {$\frac{1}{3}$} (sq2a);
\path (s2) edge [{font=\small}] node[left = 0.01cm] {$\frac{1}{3}$} (sq2b);
\path (s2) edge [{font=\small}] node[left = 0.01cm] {$\frac{1}{3}$} (sq2c);
\path (sq2a) edge [bend right=40] node {} (s1);
\path (sq2a) edge [bend left=40] node {} (s1);
\path (sq2b) edge [bend right=40] node {} (s2);
\path (sq2b) edge [bend left=40] node {} (s2);
\path (sq2c) edge [bend right=40] node {} (s2);
\path (sq2c) edge [bend left=40] node {} (s2);
\end{scope}

\begin{scope}[shift={(10cm,2cm)}]
\pgflowlevelsynccm
\resizebox{10cm}{!}{%
$
g(\begin{pmatrix} 
x_1\\
x_2\\
x_3\\
x_4\\
\end{pmatrix}) = 
\begin{pmatrix} 
\frac{1}{3} x_2 + \frac{1}{3}x_3 +\frac{1}{3}x_4 \\
\frac{1}{4}(3-\sqrt{7})\cdot \frac{1}{4}(3-\sqrt{7})\\
x_1\cdot x_1\\
x_1\cdot x_1
\end{pmatrix}
$
}
\end{scope}
\end{tikzpicture}
\caption{Language $L_3$ --- the third approximation of the language $L_\infty$. Transitions from $q_1,q_2$ are depicted in Figures \ref{gvv_1},\ref{gvv_2}. The initial state is $q_3$. The middle figure represents the MBP ${\mathcal M}_3$ corresponding to the automaton. }
\label{gvv_4}
\end{figure}
%\end{wrapfigure}

According to Theorem \ref{bellman}, the measure of language $L_3$ is equal to the value of the MBP ${\mathcal M}_3$, which 
is still simple enough, so that it can be analyzed directly by ad hoc methods. The MBP $\mathcal{M}_3$ corresponding to $\mathcal{A}_3$ extends the MBP $\mathcal{M}_2$ of the previous example with the probabilistic state $s_{q_3}$ and the three $\forall$-branching states $s_{q_3,a}$, $s_{q_3,b}$ and $s_{q_3,c}$. The new part of the automaton $\mathcal{A}_3$ is depicted in Figure \ref{gvv_2} (left). Noting the four new states are not reachable by the other eight states already present in $\mathcal{M}_2$, we already know that $val_{q_2}\!=\!\frac{1}{4}(3-\sqrt{7})$. Hence we can consider the simplified system of fixed-point equations $\mu \vec{x}.g(\vec{x})$ for calculating the values $val\!=\!(val_{q_3}, val_{q_{3},a}, val_{q_{3},b}, val_{{q_3},c})$ where $g$ is defined in Figure \ref{gvv_2} (right).
Hence the value $val_{q_3}$ is the least solution in $[0,1]$ of the equation 
\[ x = \frac{1}{3}(\frac{1}{4}(3-\sqrt{7}))^2 + \frac{2}{3}x^2, \]
which is equal to $\frac{1}{4}(3-\sqrt{1+3\sqrt{7}})$.

We will obtain the above conclusion using our computational scheme and the tool {\tt qepcad}. For this sake we will analyze the MBP ${\mathcal M}_3$ from scratch, first translating it into 
a system of 12 fixed-point equations. 
\begin{center}
$val = \begin{pmatrix} 
val_{s_{q_1}}\\
val_{s_{q_1,a}}\\
val_{s_{q_1,b}}\\
val_{s_{q_1,c}}\\
val_{s_{q_2}}\\
val_{s_{q_2,a}}\\
val_{s_{q_2,b}}\\
val_{s_{q_2,c}}\\
val_{s_{q_3}}\\
val_{s_{q_3,a}}\\
val_{s_{q_3,b}}\\
val_{s_{q_3,c}}\\
\end{pmatrix} \ \   \textnormal{ and } \ \ 
g(\begin{pmatrix} 
y_1\\
y_2\\
y_3\\
y_4\\
y_5 \\
y_6 \\
y_7 \\
y_8 \\
y_9 \\
y_{10} \\
y_{11} \\
y_{12} 
\end{pmatrix}
) = 
\begin{pmatrix} 
\frac{1}{3} y_2 +  \frac{1}{3} y_3 + \frac{1}{3} y_4  \\
1 \cdot 1 \\
y_1 \cdot y_1 \\
y_1 \cdot y_1 \\
\frac{1}{3} y_6 +  \frac{1}{3} y_7 + \frac{1}{3} y_8  \\
y_1 \cdot y_1 \\
y_5 \cdot y_5 \\
y_5 \cdot y_5 \\
\frac{1}{3} y_{10} +  \frac{1}{3} y_{11} + \frac{1}{3}  y_{12}  \\
y_5 \cdot y_5 \\
y_9 \cdot y_9 \\
y_9 \cdot y_9 \\
\end{pmatrix}
$
\end{center}
After elementary reductions and reassigning of the variables $x_1\gets y_1, x_2\gets y_5, x_3\gets y_9$, the above equation is equivalent to the following system of equations
\[\left\{
\begin{array}{lcl} 
x_1 &  \stackrel{\mu}{=} & \frac{1}{3} + \frac{2}{3}x_1^2\\
x_2 &  \stackrel{\mu}{=} & \frac{1}{3}x_1^2 + \frac{2}{3}x_2^2\\
x_3 &  \stackrel{\mu}{=} & \frac{1}{3}x_2^2 + \frac{2}{3}x_3^2
\end{array}
\right.
\]

Like in the case of language $L_2$, we start the formal analysis in {\tt qepcad} from the case of $\phi_1(x_1)$ characterizing the least solution to the first equation. 
\begin{lstlisting}
(A x1prime) 
[
    [x1 = 1/3 + 2/3 x1^2] /\ 
    [x1prime = 1/3 + 2/3 x1prime^2 ==> x1prime >= x1 ]
].
\end{lstlisting}
\vspace{-15pt}
The solver {\tt qepcad} gives us solution $2x_1 - 1 = 0$, which we use to write down the formula $\phi_2(x_2)$, describing  $x_2$ as the least solution to the second equation. 
\begin{lstlisting}
(E x1)(A x2prime)(A x1prime)
[   
    [2 x1 - 1 = 0] /\ 
    [x2 = 1/3 x1^2 + 2/3 x2^2] /\
    [x2 <= 1] /\
    [
        [
            [2 x1prime - 1 = 0] /\ 
            [x2prime <= 1] /\
            [x2prime = 1/3 x1prime^2 + 2/3 x2prime^2]
        ] ==> 
        [x2prime >= x2]
     ]
].
\end{lstlisting}
\vspace{-15pt}
As we already computed before, the formula $\phi_2(x_2)$ translates to a quantifier free expression
\begin{lstlisting}
x2 - 1 < 0 /\ 8 x2^2 - 12 x2 + 1 = 0
\end{lstlisting}
\vspace{-15pt}
Below in formula $\phi_3(x_3)$ we characterize $x_3$ as the least solution to the third equation. In the formula we use the above quantifier--free translation of the formula $\phi_2(x_2)$.  
This leads to the following formulation of $\phi_3(x_3)$.  
\begin{lstlisting}
(E x2)(A x3prime)(A x2prime)
[   
    [x2 - 1 < 0 /\ 8 x2^2 - 12 x2 + 1 = 0] /\ 
    [x3 = 1/3 x2^2 + 2/3 x3^2] /\
    [x3 <= 1] /\
    [
        [
            [x2prime - 1 < 0 /\ 8 x2prime^2 - 12 x2prime + 1 = 0] /\ 
            [x3prime <= 1] /\
            [x3prime = 1/3 x2prime^2 + 2/3 x3prime^2]
        ] ==> 
        [x3prime >= x3]
     ]
].
\end{lstlisting}
The formula $\phi_3(x_3)$ is translated by {\tt qepcad} to the following quantifier free expression
\begin{lstlisting}
x3 - 1 < 0 /\ 256 x3^4 - 768 x3^3 + 832 x3^2 - 384 x3 + 1 = 0
\end{lstlisting}

\vspace{-25pt}
\noindent
The only real solution is $x_3=\frac{1}{4}(3-\sqrt{1+3\sqrt{7}})$, which is approximately $0.0026$. This example shows in particular, that $x_3$ is not a quadratic irrational. In the case of the language $L_2$, the measure was expressible just with square roots of rational numbers. %\todo{I have no real reason to think so, but somehow find the speed of convergence as a potentially interesting topic. Here we have 1 decimal figure per iteration (so far).}

%%%%%%%%%%%%%%%%%%%%%%%%%%%%%%%%%%
%
% About L_infty. 
%
%%%%%%%%%%%%%%%%%%%%%%%%%%%%%%%%%%

\subsection{Example of a comeager language of probability $0$}
\label{subsection:linfty}
 This subsection constitutes a proof of Proposition \ref{prop_example_2}.
The regular language $L_{\infty}$ is recognized by the (deterministic) game automaton already defined in Section \ref{automata2mbp} and depicted in Figure \ref{gvv_3} (left), where the states $q_1$ and $q_2$ have priority $1$ and $2$, respectively. The MBP associated with this automaton,  depicted in Figure \ref{gvv_3} (right), has eight states. The vector of values $val$ is %the solution of the system of nested greatest and least fixed-point 
equal to $\nu \vec{y}^{\hspace{.16667em}2}. \mu \vec{y}^{\hspace{.16667em}1}.g(\vec{y}^{\hspace{.16667em}1},\vec{y}^{\hspace{.16667em}2})$ where
\begin{center}
%\resizebox{8cm}{!}{%
$val = \begin{pmatrix} 
val_{s_{q_1}}\\
val_{s_{q_1,a}}\\
val_{s_{q_1,b}}\\
val_{s_{q_1,c}}\\
val_{s_{q_2}}\\
val_{s_{q_2,a}}\\
val_{s_{q_2,b}}\\
val_{s_{q_2,c}}\\
\end{pmatrix} \ \   \textnormal{ and } \ \ 
g(\begin{pmatrix} 
y_1\\
y_2\\
y_3\\
y_4\\
\_ \\
\_ \\
\_ \\
\_ 
\end{pmatrix},
\begin{pmatrix} 
\_ \\
\_ \\
\_ \\
\_ \\
y_5 \\
y_6 \\
y_7 \\
y_8 
\end{pmatrix}
) = 
\begin{pmatrix} 
\frac{1}{3} y_2 +  \frac{1}{3} y_2 + \frac{1}{3} y_4  \\
y_5 \cdot y_5 \\
y_1 \cdot y_1 \\
y_1 \cdot y_1 \\
\frac{1}{3} y_5 +  \frac{1}{3} y_1 + \frac{1}{3} y_1  \\
y_5 \cdot y_5 \\
y_1 \cdot y_1 \\
y_1 \cdot y_1 \\
\end{pmatrix}
$
%}
\end{center}
By straightforward simplifications we obtain the system of fixed-point equations %5$\big\{ x_1  \stackrel{\mu}{=} \frac{1}{3}(x_5\cdot x_5) + \frac{2}{3}(x_1\cdot x_1) \ \ , \ \ x_5 \stackrel{\nu}{=}  \frac{1}{3}(x_5\cdot x_5) + \frac{2}{3}(x_1\cdot x_1)\big\}$ 
%The measure of language $L_\infty$ is described by the following system of equations, which was introduced in Section \ref{failvv}.
\[\left\{
\begin{array}{lcl} 
x_1 &  \stackrel{\mu}{=}  & \frac{1}{3}x_2^2 + \frac{2}{3}x_1^2\\
x_2 &  \stackrel{\nu}{=}  & \frac{1}{3}x_2^2 + \frac{2}{3}x_1^2
\end{array}
\right.
\]
in the two variables $x_1$ and $x_2$ (corresponding to the variables $y_1$, representing $s_{q_1}$, and  $y_5$, representing $s_{q_2}$) which has solution $(0,0)$, as the automated computation below will reveal. Hence $val_{s_{q_2}}\!=\! 0$ and thus we will show that the probability $\mu(L_{\infty})$ of the language $L_{\infty}$ is $0$. %\todo{You proposed some changes here in the conference version of the paper - I think they are not applicable, but please, have a critical look}

The system of fixed-point equations in the variables $x_1$ and $x_2$ can be reduced to a formula  in the language of real closed fields. First, we define a formula $\phi(x_1,x_2)$ which characterizes $x_1$ as the least fixed-point of the equation $x_1\!=\!\frac{1}{3}x_2^2 + \frac{2}{3}x_1^2$ where $x_2$ is a parameter. 
\[\begin{array}{lcl}
\psi_1(x_1,x_2) &  \bydef & (x_1 = \frac{1}{3}x_2^2 + \frac{2}{3}x_1^2)\  \wedge \ \forall x_1' . \big( (x_1' = \frac{1}{3} x_2^2 + \frac{2}{3} (x_1')^2) \Rightarrow (x_1' \geq x_1) \big) 
\end{array}
\]  
In the syntax of the tool {\tt qepcad} (see \cite{tool:qepcad} for more details) we have 
\begin{lstlisting}
(A x1prime) 
[
    [x1 = 1/3 x2^2 + 2/3 x1^2] /\ 
    [x1prime = 1/3 x2^2 + 2/3 x1prime^2 ==> x1prime >= x1 ]
].
\end{lstlisting}
\vspace{-12pt}
The solver {\tt qepcad} produces an equivalent quantifier free formula
\begin{lstlisting}
4 x1 - 3 <= 0 /\ x2^2 + 2 x1^2 - 3 x1 = 0
\end{lstlisting}

\vspace{-12pt}
We now define a formula $\psi_2(x_2)$ characterizing $x_2$ as the greatest solution to the equation $x_2\!=\!\frac{1}{3}x_2^2 + \frac{2}{3}x_1^2$ where $x_1$ is the unique number satisfying $\psi_1(x_1,x_2)$:
%We use this translation in order to write down a formula $\psi_2(x_2)$ characterizing $x_2$ as the greatest solution to the second equation.  In the
%syntax of the first order logic we have
\[\begin{array}{lcl}
\psi_2(x_2) & \bydef & \exists x_1. \big( x_2 \!=\! \frac{1}{3}x_2^2 + \frac{2}{3}x_1^2) \ \wedge\    \psi_1(x_1,x_2) \big) \ \wedge \\
            &   & \forall x_2' \big( \exists_{x'_1} (x'_2 = \frac{1}{3}(x_2')^2 + \frac{2}{3}(x_1')^2) \ \wedge \ \psi_1(x_1',x_2')\big)   \Rightarrow x_2' \leq x_2.    
\end{array}
\]
In the {\tt qepcad} syntax this can be expressed using the following code:
\begin{lstlisting}
(E x1)(A x2prime)(A x1prime)
[   
    [4 x1 - 3 <= 0 /\ x2^2 + 2 x1^2 - 3 x1 = 0] /\
    [x2 = 1/3 x2^2 + 2/3 x1^2] /\
    [
        [
            [4 x1prime - 3 <= 0 /\ x2prime^2 + 2 x1prime^2 - 3 x1prime = 0] /\ 
            [x2prime = 1/3 x2prime^2 + 2/3 x1prime^2]
        ] ==> 
        [x2prime <= x2]
     ]
].
\end{lstlisting}
\vspace{-22pt}
The solver {\tt qepcad} translates $\psi_2(x_2)$ to the following quantifier free expression
\begin{lstlisting}
x2 = 0
\end{lstlisting}
\vspace{-22pt}
This means that the measure of the language $L_\infty$ is $0$.
Our example is concluded by the following %observation (a proof can be found in Appendix \ref{app_failvv}).
\newcommand{\propmeag}{ 
The set $L_{\infty}\!\subseteq\! \trees{a,b,c}$ is comeager.} 
\begin{proposition}
\label{propmeag}
\propmeag
\end{proposition}
\begin{proof}
The proof consists of unfolding of the definition of a comeager set. From the definition it is enough to verify that the complement of $L$ is meager. Notice, that $t\not\in L$ if and only if $t\not\in L_n$ for certain $n\in\omega$. Hence, it is enough to prove that for every $n\in\omega$ the set $\Tr_{\{a,b,c\}}\setminus L_n$ is nowhere dense. That is, we have to show that for every non--empty open set $U\subset\Tr_{\{a,b,c\}}$ there exists $V\subset U$, an open subset, such that $V\subset L_n$. We can assume that $U$ is a set of trees extending certain prefix $p$ of the binary tree. As $V$ we take a set of trees which extends a prefix $q$, where $q$ is an extension of $p$ by $n$ frontiers consisting of letters $a$. 
\end{proof}

%%%%%%%%%%%%%%%%%%%%%%%%%%%%%%%%%%%%%%%%%%%%%%%%%
%
%  Computing measure of W_{i,k}
%
%%%%%%%%%%%%%%%%%%%%%%%%%%%%%%%%%%%%%%%%%%%%%%%%%
\subsection{Computing the measure of $W_{i,k}$}
\label{wik}
\label{subsec:wik}
The family of regular languages $W_{i,k}$, indexed by pairs $ i \! < \! k$ of natural numbers, constitutes a tool for investigating properties of regular languages using topological methods (\cite[p. 329]{Arnold99}, see also \cite{NiwinskiArnold2008,Bradfield98,GMMS2014}). 
The standard game automaton $\mathcal{A}_{i,k}$ over the language $\Sigma_{i,k}\!=\! \{\forall,\exists\}\times\{ i, i+1,\dots, k-1,k\}$ accepting $W_{i,k}\!\subseteq\! \trees{\Sigma_{i,k}}$ % (see also \cite{prof_fach}). 
is defined as $\mathcal{A}_{i,k}\!=\! \langle Q, q_i,\delta, \pi\rangle$ where  $Q\!=\!\{q_i,q_{i+1},\dots, q_{k}\}$, the initial state is $q_i$ and, for each $i\!\leq\! j\! \leq\! k$, the state $q_j$ has priority $\pi(q_j)\!=\!j$ and the transition function $\delta$ is defined on $q_j$ as in Figure \ref{wik_auto}. Our proof of Proposition \ref{prop_example_3}, stated in the Introduction, goes by analyzing the system of fixed-point equations associated with the game automaton $\mathcal{A}_{i,k}$. Importantly, such a system consists of linear equations and not, as in the general case, of higher order polynomials. This system can be solved using    standard   techniques of linear algebra.
%\tikzset{>={Latex[width=3mm,length=3mm]}}
%\begin{wrapfigure}[10]{R}{0.55\textwidth}
\begin{figure}[H]
\centering
%\vspace{-15pt}
\begin{tikzpicture}[scale=0.67,->,level distance=1in,sibling distance=.15in,>=stealth',thick]
%\tikzset{every tree node/.style={align=center,anchor=east,{font=\small}}}
%\tikzset{every tree child node/.style={align=center,anchor=west}}

\tikzset{grow'=up}
\Tree 
[.\node (rootprime) [draw, circle, minimum size = \diameter] {$q_j$};
        \edge[{font=\small}] node[left = 0.3cm, near end] {$\exists,i$}; 
        		[.\node [draw, diamond, aspect=2, minimum size = 2*\diameter,inner sep=1pt] {};	
						\node [draw, circle, minimum size = \diameter] {$q_i$};
						\node [draw, circle, minimum size = \diameter] {$q_i$};
                ]
        \edge[{font=\small}] node[right = 0.3cm, near end] {$\forall,i$}; 
				[.\node [draw, minimum size = 2*\smalldiameter] {};
						\node [draw, circle, minimum size = \diameter] {$q_i$};
						\node [draw, circle, minimum size = \diameter] {$q_i$};
                ]
        \edge[draw=none]; \node {};         
        \edge[draw=none]; \node {{\bf \ldots}}; 
        \edge[draw=none]; \node {}; 
        \edge[{font=\small}] node[left = 0.2cm, near end] {$\exists,k$}; 
        		[.\node [draw, diamond, aspect=2, minimum size = 2*\diameter,inner sep=1pt] {};	
						\node [draw, circle, minimum size = \diameter] {$q_k$};
						\node [draw, circle, minimum size = \diameter] {$q_k$};
                ]
        \edge[{font=\small}] node[right= 0.4cm, near end] {$\forall,k$}; 
				[.\node [draw, minimum size = 2*\smalldiameter] {};
						\node [draw, circle, minimum size = \diameter] {$q_k$};
						\node [draw, circle, minimum size = \diameter] {$q_k$};
                ]
]

\end{tikzpicture}
\caption{Transition of the automaton $\mathcal{A}_{i,k}$ recognizing $W_{i,k}$.}
\label{wik_auto}
\end{figure}
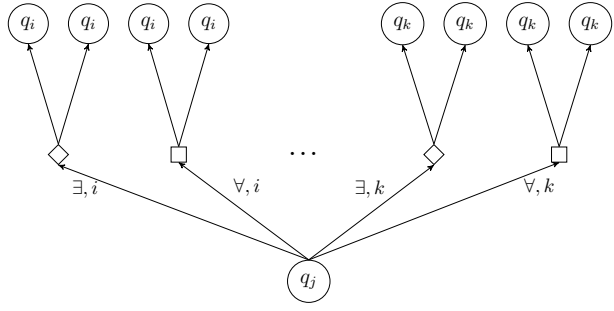
%\end{wrapfigure}

\newcommand{\wikdet}{For every $k\geq 3$ holds \[ \det(A_k) = (-1)^{k-1}\frac{1}{k}. \] } 
\begin{proof} (of Proposition \ref{prop_example_3})
In what follows we use notation $x\odot y\!=\! x+y-xy$ for denoting binary coproducts. Let ${\mathcal A}$ be the automaton  depicted in Figure \ref{wik_auto} accepting the language $W_{i,k}$ ($k>i$, $i=0,1$). The automaton has $k-i$ states $q_i,q_{i+1},\ldots,q_k$. We just consider the case $i\!=\!1$ and $k$ odd, since the other cases are analogous. We will show that the measure of $W_{1,k}$ is $0$. From the definition of the game automaton ${\mathcal A}$ we infer the corresponding system of $k$ fixed--point equations:

\[\left\{
\begin{array}{lcl} 
x_1 & \stackrel{\mu}{=} & \frac{1}{2k}(x_1\odot x_1 + x_1\cdot x_1 + \ldots + x_k\odot x_k + x_k\cdot x_k)\\
x_2 & \stackrel{\nu}{=} & \frac{1}{2k}(x_1\odot x_1 + x_1\cdot x_1 + \ldots + x_k\odot x_k + x_k\cdot x_k)\\
& \vdots & \\
x_{k-1} & \stackrel{\nu}{=} & \frac{1}{2k}(x_1\odot x_1 + x_1\cdot x_1 + \ldots + x_k\odot x_k + x_k\cdot x_k)\\
x_{k} & \stackrel{\mu}{=} & \frac{1}{2k}(x_1\odot x_1 + x_1\cdot x_1 + \ldots + x_k\odot x_k + x_k\cdot x_k)
\end{array}
\right.
\]
Note that the first and last equations are least-fixed point equations since $1$ and $k$ are odd by assumption.  Unfolding the definition of $\odot$ we get
\[\left\{
\begin{array}{lcl} 
x_1 & \stackrel{\mu}{=} & \frac{1}{2k}((2x_1-x_1^2 + x_1^2) + \ldots + (2x_k-x_k^2 + x_k^2))\\
x_2 & \stackrel{\nu}{=} & \frac{1}{2k}((2x_1-x_1^2 + x_1^2) + \ldots + (2x_k-x_k^2 + x_k^2))\\
& \vdots & \\
x_k & \stackrel{\mu}{=} & \frac{1}{2k}((2x_1-x_1^2 + x_1^2) + \ldots + (2x_k-x_k^2 + x_k^2))
\end{array}
\right.
\]
This simplifies to the following system of equations
\[\left\{
\begin{array}{lcl} 
x_1 & \stackrel{\mu}{=} & \frac{1}{2k}(2x_1 + \ldots + 2x_k)\\
x_2 & \stackrel{\nu}{=} & \frac{1}{2k}(2x_1 + \ldots + 2x_k)\\
& \vdots & \\
x_k & \stackrel{\mu}{=} & \frac{1}{2k}(2x_1 + \ldots + 2x_k)
\end{array}
\right.
\]
%The formula defining the fixed--point expression starts with $\exists_{x_k}\ \ldots$. 
Since $k$ is odd, we are looking for the least $x_k\in [0,1]$ satisfying the fixed--point formula. We will show that $0$ is a fixed-point and, therefore, the desired least fixed-point. Substituting $0$ for $x_k$ we get:
\[\left\{
\begin{array}{lcl} 
x_1 & \stackrel{\mu}{=} & \frac{1}{k}(x_1 + \ldots + x_{k-1})\\
x_2 & \stackrel{\nu}{=} & \frac{1}{k}(x_1 + \ldots + x_{k-1})\\
& \vdots & \\
x_{k-1} & \stackrel{\nu}{=}  & \frac{1}{k}(x_1 + \ldots + x_{k-1})\\
0 & = & \frac{1}{k}(x_1 + \ldots + x_{k-1})
\end{array}
\right.
\]

We are going to show that the solution of this system of fixed-point equations is the vector $x_1=0,x_2=0\ldots,x_{k-1}=0$. This will conclude the proof of Proposition \ref{prop_example_3}.

To prove this, note that the solution is necessarily a solution of the linear system of equations
\[\left\{
\begin{array}{lcl} 
x_1 & = & \frac{1}{k}(x_1 + \ldots + x_{k-1})\\
x_2 & = & \frac{1}{k}(x_1 + \ldots + x_{k-1})\\
& \vdots & \\
x_{k-1} & =  & \frac{1}{k}(x_1 + \ldots + x_{k-1})\\
0 & = & \frac{1}{k}(x_1 + \ldots + x_{k-1})
\end{array}
\right.
\]
which, as we now prove, has only one solution. To prove this we now focus on the first $k-1$ equations and show that even without the last equation the only solution is $x_1=0,\ldots,x_{k-1}=0$. This vector of numbers satisfies also the last equation $0 = \frac{1}{k}(x_1 + \ldots + x_{k-1})$. 
In the first step we move the left--hand side variables to 
the right--hand side. 
\[\left\{
\begin{array}{lcl} 
0 & = & -\frac{k-1}{k} x_1 + \frac{1}{k} x_2 + \ldots + \frac{1}{k} x_{k-1}\\
0 & = & \frac{1}{k}x_1 - \frac{k-1}{k} x_2 \ldots + \frac{1}{k} x_{k-1}\\
& \vdots & \\
0 & = & \frac{1}{k}x_1 + \ldots + \frac{1}{k}x_{k-1} - \frac{k-1}{k} x_k
\end{array}
\right.
\]
Hence we are getting a $(k-1)\times (k-1)$ matrix $A_k$ with the following coefficients.
\[ A_{k} = 
\begin{pmatrix} 
-\frac{k-1}{k}  & \frac{1}{k}       &   \ldots  & \frac{1}{k} \\
\frac{1}{k}     & -\frac{k-1}{k}    &   \ldots  & \frac{1}{k} \\
& \vdots & \\
\frac{1}{k}     & \frac{1}{k}       &   \ldots  & -\frac{k-1}{k}
\end{pmatrix}
\]
%In other terms, the matrix has everywhere coefficients $\frac{1}{k}$, except of the diagonal, where coefficients are $-\frac{k-1}{k}$. 
%Clearly, $x_1=0,x_2=0,\ldots,x_{k-1}=0$ is a solution of the system of equations and it is enough to check that this is the only solution. 

This is a linear system of equations and therefore it is enough to verify that the determinant of the matrix is different than $0$.  
Hence the Lemma below finishes the proof of Proposition \ref{prop_example_3}.
\end{proof}

\begin{lemma}\label{lemma:wikdet} 
\wikdet
\end{lemma}
%We present a simple computational proof of this Lemma in the Appendix \ref{app_wik}. 
\noindent
Before the proof let us consider a special of $k=3$ . Then $ A_{3} = 
\begin{pmatrix} 
-\frac{2}{3}  & \frac{1}{3} \\
\frac{1}{3} & -\frac{2}{3}
\end{pmatrix}$, hence $\det(A_3)=\frac{4}{9}-\frac{1}{9}=\frac{3}{9}=\frac{1}{3}$. Moreover, one can notice that if we do not limit attention to the case $x_k=0$, then the system of equations has infinitely many solutions, because the matrix $\begin{pmatrix} 
-\frac{2}{3}  & \frac{1}{3} & \frac{1}{3} \\
\frac{1}{3} & -\frac{2}{3} & \frac{1}{3} \\
\frac{1}{3} & \frac{1}{3} & -\frac{2}{3}
\end{pmatrix}$ has determinant equal to $0$.

%\label{comp_wik}
%In this part of the Appendix we prove
%\begin{replemma}{lemma:wikdet}
%\wikdet
%\end{replemma}

\begin{proof}(of Lemma \ref{lemma:wikdet})

We simplify the determinant of the $(k-1)\times (k-1)$ matrix $A_k$ in order to have $0$'s in the first row on all positions, with except of the first one. 
\[\det(A_k) = \det
\begin{pmatrix} 
-\frac{k-1}{k}  & \frac{1}{k}       &   \frac{1}{k} & \ldots & \frac{1}{k} \\
\frac{1}{k}     & -\frac{k-1}{k}    &   \frac{1}{k} & \ldots & \frac{1}{k} \\
& \vdots & \\
\frac{1}{k}     & \frac{1}{k}       &   \frac{1}{k} & \ldots & -\frac{k-1}{k}
\end{pmatrix} \]
\[= \det
\begin{pmatrix} 
(-\frac{k-1}{k}-\frac{1}{k}) & (\frac{1}{k}+\frac{k-1}{k})   &  0 & \ldots & 0 \\
\frac{1}{k}     & -\frac{k-1}{k}    &   \frac{1}{k} & \ldots & \frac{1}{k} \\
& \vdots & \\
\frac{1}{k}     & \frac{1}{k}       &   \frac{1}{k} & \ldots & -\frac{k-1}{k}
\end{pmatrix} \]
\[ = \det
\begin{pmatrix} 
-1  & 1       &   0 & \ldots & 0 \\
\frac{1}{k}     & -\frac{k-1}{k}    &   \frac{1}{k} & \ldots & \frac{1}{k} \\
& \vdots & \\
\frac{1}{k}     & \frac{1}{k}       &   \frac{1}{k} & \ldots & -\frac{k-1}{k}
\end{pmatrix} \]
\[= \det
\begin{pmatrix} 
-1  & 1-1       &   0 & \ldots & 0 \\
\frac{1}{k}     & (-\frac{k-1}{k}+\frac{1}{k})    &   \frac{1}{k} & \ldots & \frac{1}{k} \\
& \vdots & \\
\frac{1}{k}     & (\frac{1}{k}+\frac{1}{k})       &   \frac{1}{k} & \ldots & -\frac{k-1}{k}
\end{pmatrix} \]
\[= \det
\begin{pmatrix} 
-1  & 0       &   0 & \ldots & 0 \\
\frac{1}{k}     & -\frac{k-2}{k}    &   \frac{1}{k} & \ldots & \frac{1}{k} \\
& \vdots & \\
\frac{1}{k}     & \frac{2}{k}       &   \frac{1}{k} & \ldots & -\frac{k-1}{k}
\end{pmatrix}.
\]

Having $0$'s in the first row we unfold the determinant with respect to this row. The next determinant 
is of the size $(k-2)\times (k-2)$. Again, we aim to have $0$'s in the first row on all positions, with except of the first one. 

\[(-1) \det
\begin{pmatrix} 
-\frac{k-2}{k} & \frac{1}{k}  &  \frac{1}{k} & \ldots & \frac{1}{k} \\
\frac{2}{k}     & -\frac{k-1}{k}    &   \frac{1}{k} & \ldots & \frac{1}{k} \\
& \vdots & \\
\frac{2}{k}     & \frac{1}{k}       &   \frac{1}{k} & \ldots & -\frac{k-1}{k}
\end{pmatrix} \]
\[ = (-1) \det
\begin{pmatrix} 
-1  & 1       &   0 & \ldots & 0 \\
\frac{2}{k}     & -\frac{k-1}{k}    &   \frac{1}{k} & \ldots & \frac{1}{k} \\
& \vdots & \\
\frac{2}{k}     & \frac{1}{k}       &   \frac{1}{k} & \ldots & -\frac{k-1}{k}
\end{pmatrix} \]
\[= (-1) \det
\begin{pmatrix} 
-1  & 1-1       &   0 & \ldots & 0 \\
\frac{2}{k}     & (-\frac{k-1}{k}+\frac{2}{k})    &   \frac{1}{k} & \ldots & \frac{1}{k} \\
& \vdots & \\
\frac{2}{k}     & (\frac{1}{k}+\frac{2}{k})       &   \frac{1}{k} & \ldots & -\frac{k-1}{k}
\end{pmatrix} \]
\[ = (-1) \det
\begin{pmatrix} 
-1  & 0       &   0 & \ldots & 0 \\
\frac{2}{k}     & -\frac{k-3}{k}    &   \frac{1}{k} & \ldots & \frac{1}{k} \\
& \vdots & \\
\frac{2}{k}     & \frac{3}{k}       &   \frac{1}{k} & \ldots & -\frac{k-1}{k}
\end{pmatrix}.
\]

Now we are ready to again unfold the determinant with respect to the first row. In the following computations we deal 
with the matrix of the size $(k-3)\times (k-3)$. 

\[(-1)^2 \det
\begin{pmatrix} 
-\frac{k-3}{k} & \frac{1}{k}  &  \frac{1}{k} & \ldots & \frac{1}{k} \\
\frac{3}{k}     & -\frac{k-1}{k}    &   \frac{1}{k} & \ldots & \frac{1}{k} \\
& \vdots & \\
\frac{3}{k}     & \frac{1}{k}       &   \frac{1}{k} & \ldots & -\frac{k-1}{k}
\end{pmatrix} \]
\[ = (-1)^2 \det
\begin{pmatrix} 
-1  & 1       &   0 & \ldots & 0 \\
\frac{3}{k}     & -\frac{k-1}{k}    &   \frac{1}{k} & \ldots & \frac{1}{k} \\
& \vdots & \\
\frac{3}{k}     & \frac{1}{k}       &   \frac{1}{k} & \ldots & -\frac{k-1}{k}
\end{pmatrix} \]
\[= (-1)^2 \det
\begin{pmatrix} 
-1  & 1-1       &   0 & \ldots & 0 \\
\frac{3}{k}     & (-\frac{k-1}{k}+\frac{3}{k})    &   \frac{1}{k} & \ldots & \frac{1}{k} \\
& \vdots & \\
\frac{3}{k}     & (\frac{1}{k}+\frac{3}{k})       &   \frac{1}{k} & \ldots & -\frac{k-1}{k}
\end{pmatrix} \]
\[= (-1)^2 \det
\begin{pmatrix} 
-1  & 0       &   0 & \ldots & 0 \\
\frac{3}{k}     & -\frac{k-4}{k}    &   \frac{1}{k} & \ldots & \frac{1}{k} \\
& \vdots & \\
\frac{3}{k}     & \frac{4}{k}       &   \frac{1}{k} & \ldots & -\frac{k-1}{k}
\end{pmatrix} = \cdots
\]

Finally, after $k-3$ steps, we are ending up with a determinant of the size $2\times 2$:

\[
(-1)^{k-3} \det
\begin{pmatrix} 
-1  & 1  \\
\frac{k-2}{k}     & -\frac{k-1}{k} 
\end{pmatrix} = (-1)^{k-3}(\frac{k-1}{k}-\frac{k-2}{k}) = (-1)^{k-3}\frac{1}{k}.
\]

\end{proof}

%%%%%%%%%%%%%%%%%%%%%%%%%%%%%%%%%%%%%%%%%%%%%%%%%
%
%  Migrated appendix - useless, hence finally incorporated into the main body
%
%%%%%%%%%%%%%%%%%%%%%%%%%%%%%%%%%%%%%%%%%%%%%%%%%

% \newpage 

% \appendix

%\noindent
%\section{Monadic Second Order Logic}
%\label{mso_intro}
%\input{app_tech.tex} 

% \noindent
% \section{Computations related to Examples \ref{prop_example_1} and \ref{prop_example_2}}
%\label{app_failvv}
%\input{failvv_app3.tex} 

%\section{\bf Computations for languages $W_{i,k}$}
%\section{
%\label{app_wik}
% \input{wik_app.tex}
%\input{comp_wik.tex} 

%%%%%%%%%%%%%%%%%%%%%%%%%%%%%%%%%%%%%%%%%%%%%%%%%
%
%  Conclusions
%
%%%%%%%%%%%%%%%%%%%%%%%%%%%%%%%%%%%%%%%%%%%%%%%%%
\section{Conclusion}
\label{conclusion}
%In the future work we would like to find an algorithm for computing a measure of a given regular language of trees. Since regular sets of trees are implicitly defined with fixed points, one approach to this problem is to unfold the definition of fixed points and calculate the measure by induction. This is the method used in this paper. Alternatively, 
In this work we presented an algorithm for computing the probability of regular languages defined by game automata. The Probability Problem in its full generality remains open. A possible direction for future research is to investigate approximations of regular languages by simpler regular languages. For example, given a regular language $L$ of trees, is it possible to find a regular language $G$ defined by a game automaton 
%\footnote{Another version of this question may ask for a game language $D$ instead of a deterministic one.} 
such that $L\triangle G = (L\setminus G)\cup (G\setminus L)$ is of probability $0$, i.~e.~$L$ differs from $G$ by a set of probability $0$? An effective answer to this question, that is an algorithm constructing a language $G$ from $L$, combined with the algorithm described in this paper would lead to a full solution to the Probability Problem. 

%\section*{Acknowledgment}
%\vspace{-20pt}
%\bibliographystyle{abbrv}
\bibliography{biblio}

%%%%%%%%%%%%%%%%%%%%%%%%%%%%%%%%%%%%%%%%%%%%%%%%%
%
%  Pictorial representation of MBP
%
%%%%%%%%%%%%%%%%%%%%%%%%%%%%%%%%%%%%%%%%%%%%%%%%%
%\newpage
%\section{Pictorial representation of tree generation process for a given MBP}
%\label{picgen}
%An example of MBP is illustrated below, by depicting probabilistic states, $\exists$-branching and $\forall$-branching states as circles, diamond and boxes, respectively. 
%\input{pictures/matteo_pic1.tex}
%Some possible\todo{Possible in what sense?} initial steps of the stochastic process having probabilities $(1,\frac{1}{3},1, \frac{2}{3}, \frac{1}{2},1)$, respectively, are depicted below. Since, each leaf is %extended independently from the others, the probability of the process beginning as depicted is $\frac{2}{18}$.
%\input{pictures/matteo_pic3.tex}

%%%%%%%%%%%%%%%%%%%%%%%%%%%%%%%%%%%%%%%%%%%%%%%%%
%
%  Straighthforward generalization fails
%
%%%%%%%%%%%%%%%%%%%%%%%%%%%%%%%%%%%%%%%%%%%%%%%%%
%\section{Failure of correctness of the naive generalization}
%\todo{Remove $\mu$}
%\input{failstraight}
\end{document}